\documentclass[12pt]{article}

\usepackage{fullpage}
\usepackage{graphicx}
\usepackage{amsmath,amssymb}
\usepackage{fancyhdr,url}
\usepackage{calc}
\usepackage{gastex}
\usepackage{xspace}
\usepackage{listings}
\usepackage[vlined, linesnumbered]{algorithm2e}
\usepackage{microtype}
\usepackage{bm}  
\usepackage{amsfonts, amssymb,nccmath}
\usepackage[draft]{commenting}

\usepackage{tikz}
\usetikzlibrary{arrows,automata,shapes,positioning}
\tikzstyle{every picture}+=[>=stealth,initial text=, node distance=2cm]

\usepackage{graphicx,caption,subcaption,pslatex}

\newcommand{\safety}[1]{\ensuremath{\mathsf{SAFE}_{#1}}\xspace}
\newcommand{\reachability}[1]{\ensuremath{\mathsf{REACH}_{#1}}\xspace}

\newcommand{\denote}[1]{\ensuremath{\left[\!\left[#1\right]\!\right]}}

\newcommand{\revert}[1]{\overleftarrow{#1}}

\newcommand{\asynctrans}{\ensuremath{\rightarrow_{\textit{as}}}}

\newcommand{\ViewSet}[2]{\ensuremath{\mathcal{V}_{#2,#1}}}

\newcommand{\nat}{\mathbb{N}}
\newcommand{\zed}{\mathbb{Z}}
\newcommand{\tuple}[1]{\langle #1 \rangle}
\newcommand{\set}[1]{\{#1\}}
\newcommand{\config}{\ensuremath{\mathbf{p}}}
\newcommand{\ViewR}[2]{\ensuremath{\mathbf{V}_{#1}[#2\rightarrow]}}
\newcommand{\ViewL}[2]{\ensuremath{\mathbf{V}_{#1}[\leftarrow#2]}}

\newcommand{\aterm}{\mathtt{t}}

\newcommand{\mov}{\mathit{move}}

\newcommand{\modulo}{\odot}

\newcommand{\semtrans}{\hookrightarrow}
\newcommand{\synchtrans}{\Rightarrow}
\newcommand{\asynchtrans}{\leadsto}

\newcommand{\runs}[2]{\ensuremath{\mathbf{Runs}_{#2}(#1)}}
\newcommand{\lo}{\mathbf{L}}
\newcommand{\mo}{\mathbf{M}}
\newcommand{\internalState}{\mathbf{s}}
\newcommand{\Views}{\mathbf{V}}

\newcommand{\ssync}{\mathrm{ss}}
\newcommand{\async}{\mathrm{as}}
\newcommand{\sync}{\mathrm{s}}
\newcommand{\post}{\textbf{Post}}
\newcommand{\poststar}[1]{\mathbf{Post}_{#1}^*}

\newcommand{\actif}[1]{\ensuremath{\textrm{Act}(#1)}}

\newcommand{\free}{\ensuremath{\mathsf{F}}\xspace}
\newcommand{\robot}{\ensuremath{\mathsf{R}}\xspace}
\newcommand{\block}[1]{\ensuremath{\mathsf{B}_{#1}}\xspace}
\newcommand{\stable}{\ensuremath{\mathsf{stable}}\xspace}
\newcommand{\moving}{\ensuremath{\mathsf{moving1}}\xspace}
\newcommand{\movingtwo}{\ensuremath{\mathsf{moving2}}\xspace}
\newcommand{\movingthree}{\ensuremath{\mathsf{moving3}}\xspace}
\newcommand{\stabili}{\ensuremath{\mathsf{stabilizing1}}\xspace}
\newcommand{\stabilitwo}{\ensuremath{\mathsf{stabilizing2}}\xspace}
\newcommand{\robots}{\ensuremath{\mathcal{R}}}
\newcommand{\halt}{\ensuremath{\mathsf{HALT}}}

\newcommand{\configview}[1]{\mathtt{ConfigView}_{#1}}
\newcommand{\viewsym}{\mathtt{ViewSym}}
\newcommand{\presburgmove}{\mathtt{Move}}
\newcommand{\semisyncpost}{\mathtt{SemiSyncPost}}
\newcommand{\syncpost}{\mathtt{SyncPost}}
\newcommand{\asyncpost}{\mathtt{AsyncPost}}

  \newenvironment{proofsketch}{%
  \noindent {\it Sketch of proof.}\hspace*{.3cm}
}{%
   
 }

 \newtheorem{corollary}{Corollary}{\bfseries}{\itshape}
\newtheorem{definition}{Definition}{\bfseries}{\itshape}
\newtheorem{lemma}{Lemma}{\bfseries}{\itshape}
{\bfseries}{\itshape}
\newtheorem{theorem}{Theorem}{\bfseries}{\itshape}
{\itshape}{\rmfamily}
\newenvironment{proof}{\noindent \textit{Proof:}}{\newline$\Box$}
\newtheorem{myclaim}{Claim}{\bfseries}{\itshape}

\newcommand{\qed}{\hfill$\Box$}

\newif\iflong
\longtrue
\longfalse

\begin{document}
 \title{Parameterized Verification of Algorithms for Oblivious Robots on a Ring}
 \author{Arnaud Sangnier$^*$ \and Nathalie Sznajder$^\star$ \and Maria Potop-Butucaru$^\star$ \and  S\'ebastien Tixeuil$^\star$}
 \date{$^*$ IRIF - Univ Paris Diderot - Paris, France \\
$^\star$ Sorbonne Universit\'es, UPMC Univ Paris 06, UMR 7606, LIP6, F-75005, Paris, France}

\declareauthor{ts}{Tali}{red}
\declareauthor{as}{Arnaud}{blue}

\pagestyle{plain}

\maketitle

\begin{abstract}
We study verification problems for autonomous swarms of mobile robots that
self-organize and cooperate to solve global objectives. In particular, we focus 
in this paper on the model proposed by Suzuki and Yamashita of anonymous robots
evolving in a discrete space with a finite number of locations (here, a ring). A large number
of algorithms have been proposed working for rings whose size is not a priori fixed and can be hence considered as a parameter. Handmade correctness proofs of these algorithms have been shown to be error-prone,
and recent attention had been given to the application of formal methods to automatically
prove those. Our work is the first to study the verification problem of such algorithms in the
parameterized case. We show that safety and reachability problems are 
undecidable for robots evolving asynchronously. On the positive side, we show that safety 
properties are decidable in the synchronous case, as well as in the asynchronous case for a
particular class of algorithms. Several properties on the protocol can be decided as well. Decision
procedures rely on an encoding in Presburger arithmetics formulae that can be verified by an
SMT-solver. Feasibility of our approach is demonstrated by the encoding of several case studies.
\end{abstract}

%
\section{Introduction}

We consider sets of mobile oblivious robots evolving in a discrete space (modeled as a ring shaped graph). For our purpose, rings are seen as discrete graphs whose vertices represent the different positions available to host a robot, and edges model the possibility for a robot to move from one position to another.
Robots follow the seminal model by Suzuki and Yamashita~\cite{SuzukiYamashita99}: they do not remember their past actions, they cannot communicate explicitly, and are disoriented.

However, they can sense their environment and detect the positions of the other robots on the ring. If several robots share the same position on the ring (forming a \emph{tower}, or multiplicity point), other robots may or may not detect the tower. If robots have \emph{weak} multiplicity detection, they are assumed to sense a tower on a position, but are not able to count the actual number of robots in this tower. With \emph{strong} multiplicity detection, they are able to count the exact number of robots on a given position. In case they have no multiplicity detection, robots simply detect occupied positions. In this paper, we assume strong multiplicity detection is available to all robots. 

Robots are anonymous and execute the same deterministic algorithm to achieve together a given objective. Different objectives for ring shaped discrete spaces have been studied in the literature~\cite{flocchini12book}: gathering -- starting from any initial configuration, all the robots must gather on the same node, not known beforehand, and then stop~\cite{KranakisKrizancMarkou10}, exploration with stop -- starting from any initial configuration, the robots reach a configuration where they all are idle and, in the meanwhile, all the positions of the ring have been visited by a robot~\cite{FlocchiniIPS13}, exclusive perpetual exploration -- starting from any tower-free configuration, each position of the ring is visited infinitely often and no multiplicity point ever appears~\cite{BlinMPT10,DAngeloSNNS13}. 

Each robot behaves according to the following cycle: it takes a snapshot of its environment, then it computes its next move (either stay idle or move to an adjacent node in the ring), and at the end of the cycle, it moves according to its computation. Such a cycle is called a look-compute-move cycle.

Since robots cannot rely on a common sense of direction, directions that are computed in the compute phase are only \emph{relative} to the robot. To tell apart its two sides, a robot relies on a description of the ring in both clockwise and counter-clockwise direction, which gives it two views of the configuration. There are two consequences to this fact. First, if its two views are identical, meaning that the robot is on an axis of symmetry, it cannot distinguish the two directions and thus either decides to stay idle, or to move. In the latter case, the robot moves becomes a non-deterministic choice between the two available directions. Second, when two robots have the same
two views of the ring, the protocol commands them to move in the same relative direction, but this might result in moves in actual opposite directions for the two robots. Such a symmetrical situation  is pictured in Figure~\ref{fig:symmetry}.

\begin{figure}
\centering
\begin{subfigure}[b]{0.30\textwidth}
\captionsetup{justification=centering}
\centering
\begin{tikzpicture}[node distance =10em,>=stealth,->, scale=0.6, transform shape, thick]
  \node[circle,minimum size =8.5em, draw, thick] (C0){};

  \foreach \i in {0,...,6}
    \node[circle,minimum size =1.6em,draw,fill=white,thick]at(\i*51:1.5) (\i){};

\node[circle,minimum size =1.5em,fill=gray!100]at(2*51:1.5) (4){};
\draw node[above left of=4, node distance=1.5em] {$\robot_1$};
\node[circle,minimum size =1.5em,fill=gray!100]at(5*51:1.5) (6){};
\draw node[below left of=6, node distance=1.5em] {$\robot_2$};
\node[circle,minimum size =1.5em,fill=gray!100]at(6*51:1.5) (9){};
\draw node[below right of=9, node distance=1.75em] {$\robot_3$};

\end{tikzpicture}
\caption{\small {A disoriented robot $\robot_1$}}
\label{fig:views1}
\end{subfigure}

\begin{subfigure}[b]{0.30\textwidth}
\captionsetup{justification=centering}
\centering
\begin{tikzpicture}[node distance =10em,>=stealth,->, scale=0.6, transform shape, thick]
  \node[circle,minimum size =8.5em, draw, thick] (C0){};

  \foreach \i in {0,...,9}
    \node[circle,minimum size =1.6em,draw,fill=white,thick]at(\i*36:1.5) (\i){};

\node[circle,minimum size =1.5em,fill=gray!100]at(5*36:1.5) (5){};
\node[circle,minimum size =1.5em,fill=gray!100]at(4*36:1.5) (4){};
\node[circle,minimum size =1.5em,fill=gray!100]at(2*36:1.5) (2){};
 \node[circle,minimum size =1.6em,draw,fill=white,thick]at(2*36:2) (2){};
\node[circle,minimum size =1.5em,fill=gray!100]at(2*36:2) (2){};
\draw node[above left of=4, node distance=1.5em] {$\robot_5$};
\draw node[above of=2, node distance=1.5em] {$\robot_1,\robot_2$};
\draw node[above left of=5, node distance=1.5em] {$\robot_4$};
\node[circle,minimum size =1.5em,fill=gray!100]at(9*36:1.5) (9){};
\draw node[below right of=9, node distance=2.25em] {$\robot_3$};

\end{tikzpicture}
\caption{A configuration with a tower}
\label{fig:views3}
  \end{subfigure}

\caption{}\label{fig:symmetry}
\end{figure}

Existing execution models consider different types of synchronization for the robots: in the fully synchronous model (FSYNC), all robots evolve simultaneously and complete a full look-compute-move cycle. The semi-synchronous model (SSYNC) consider runs that evolve in phases: at each phase, an arbitrary subset of the robots is scheduled for a full look-compute-move cycle, which is executed simultaneously by all robots of the subset. Finally, in the asynchronous model (ASYNC), robots evolve freely at their own pace: In particular, a robot can move according to a computation based on an obsolete observation of its environment, as others robots may have moved in between. Algorithms in the literature are typically parameterized by the number of robots and/or number of positions in the ring. In this work we focus on formally verifying algorithms parameterized by the number of ring positions only, assuming a a fixed number of robots.

\subsection{Related work}

Designing and proving mobile robot protocols is notoriously difficult. Formal methods encompass a long-lasting path of research that is meant to overcome errors of human origin. Unsurprisingly, this mechanized approach to protocol correctness was successively used in the context of mobile robots~\cite{bonnet14wssr,devismes12sss,berard16dc,auger13sss,MPST14c,courtieu15ipl,berard15infsoc,AminofMuranoRubinZuleger15,balabonski16sss}.

When robots are \emph{not} constrained to evolve on a particular topology (but instead are allowed to move freely in a bidimensional Euclidian space), the Pactole (\texttt{http://pactole.lri.fr}) framework has been proven useful. Developed for the Coq proof assistant, Pactole enabled the use of high-order logic to certify impossibility results~\cite{auger13sss} for the problem of convergence: for any positive $\epsilon$, robots are required to reach locations that are at most $\epsilon$ apart. Another classical impossibility result that was certified with Pactole is the impossibility of gathering starting from a bivalent configuration~\cite{courtieu15ipl}. Recently, positive certified results for SSYNC gathering with multiplicity detection~\cite{courtieu16disc}, and for FSYNC gathering without multiplicity detection~\cite{balabonski16sss} were provided. However, as of now, no Pactole library is dedicated to robots that evolve on discrete spaces.

In the discrete setting that we consider in this paper, model-checking proved useful to find bugs in existing literature~\cite{berard16dc,doan16sofl} and assess formally published algorithms~\cite{devismes12sss,berard16dc,AminofMuranoRubinZuleger15}. Automatic program synthesis (for the problem of perpetual exclusive exploration in a ring-shaped discrete space) is due to Bonnet \emph{et al.}~\cite{bonnet14wssr}, and can be used to obtain automatically algorithms that are ``correct-by-design''. The approach was refined by Millet \emph{et al.}~\cite{MPST14c} for the problem of gathering in a discrete ring network. As all aforementioned approaches are designed for a bounded setting where both the number of locations and the number of robots are known, they cannot permit to establish results that are valid for any number of locations. 

Recently, Aminof \emph{et al.} \cite{AminofMuranoRubinZuleger15} presented a general framework for verifying properties about mobile robots evolving on graphs, where the graphs are a parameter of the problem. 
While our model could be encoded in their framework, their undecidability proof relies on persistent memory used by the robots, hence is not applicable to the case of oblivious robots we consider here. Also, they obtain decidability in a subcase that is not relevant for robot protocols like those we consider. Moreover, their decision procedure relies on MSO satisfiability, which does not enjoy good complexity properties and cannot be implemented efficiently for the time being. 

\subsection{Contributions}
In this work, we tackle the more general problem of verifying protocols for swarms of robots for any number of locations. 

We provide a formal definition of the problem, where the protocol can be described as a quantifier free Presburger formula. This logic, weak enough to be
decidable, is however powerful enough to express existing algorithms in the literature. Objectives of the robots are also described by Presburger formulae and
we consider two problems: when the objective of the robots is a safety objective -- robots have to avoid the configurations described by the formula (\safety{}), 
and when it is a reachability objective (\reachability{}). We show that if \reachability{} is undecidable in any semantics, \safety{} is decidable in FSYNC and SSYNC.
We also show that when the protocol is \emph{uniquely-sequentializable}, safety properties become decidable even in the asynchronous case. 

Finally, we show practical applicability of this approach by using an SMT-solver to verify safety properties for some algorithms from the literature. 

Hence, we
advocate that our formalism should be used when establishing such protocols, as a formal and non-ambiguous description, instead of the very informal and
sometimes unclear definitions found in the literature. Moreover, if totally automated verification in the parameterized setting seems unfeasible, our method could be used as a ``sanity check'' of the protocol, and to automatically prove intermediate lemmas, that can then be used as formally proved building blocks of a handmade 
correction proof.

Due to lack of space, some proofs are given in the appendices.

\section{Model of Robots Evolving on a Ring}

\subsection{Formal model}
In this section we present the formal language to describe mobile robots protocols as well as the way it is interpreted.

\subsubsection{Preliminaries}

For $a,b \in \mathbb{Z}$ such that $a \leq b$, we denote by $[a,b]$
the set $\set{c \in \mathbb{Z} \mid a \leq c \leq b}$. For $a \in
\zed$ and $b \in
\nat$, we write $a \modulo b$ the natural $d \in [0,(b-1)]$
such that there exists $j \in \zed$ and $a=b.j+d$ (for instance $-1
\modulo 3=2$). Note that $\modulo$ corresponds to the modulo operator, but for sake of clarity we recall its definition when $a$ is negative.

We recall the definition of Existential Presburger  (EP) formulae. Let $Y$ be a countable set of variables.  First we define the
grammar for terms $\aterm ::= x ~\mid~ \aterm + \aterm ~\mid~ a\cdot \aterm ~\mid~ \aterm \mod a$, where $a\in
\nat$ and $x \in Y$ and then the
grammar for formulae is given by $\phi ::= \aterm \bowtie b ~\mid~
\phi \wedge \phi ~\mid~ \phi \vee \phi ~\mid~ \exists x. \phi$ where ${\bowtie} \in
\set{=,\leq,\geq,<,>}$, $x \in Y$ and $b\in\nat$. We sometimes write a formula $\phi$ as $\phi(x_1,\ldots,x_k)$ to underline that $x_1,\ldots,x_k$ are the free variables of $\phi$. The set of Quantifier Free Presburger (QFP) formulae is obtained by the same grammar deleting the elements $\exists x. \phi$.
 Note that when dealing with QFP formulae, we allow as well negations of formulae.

We say that a vector
$V=\tuple{d_1,\dots,d_k}$ satisfies  an EP formula $\phi(x_1,\ldots,x_k)$, denoted
by $V \models \phi$, if the formula obtained by replacing each
$x_i$ by $d_i$ holds.
Given a formula $\phi$ with free variables ${x_1,\dots, x_k}$, we write
$\phi(d_1,\dots,d_k)$ the formula where each $x_i$ is replaced by $d_i$.
We let $\denote{\phi(x_1,\ldots,x_k)}=\{\tuple{d_1,\dots,d_k}\in \mathbb{N}^k\mid 
 \phi(x_1,\ldots,x_k)\models\phi\}$ be the set of models of the formula.
In the sequel, we use Presburger formulae to describe configurations
 of the robots, as well as protocols.

\subsubsection{Configurations and robot views}
In this paper, we consider a fixed number $k >0$ of robots  and, except when stated otherwise, we assume the identities
of the robots are $\robots=\set{R_1,\ldots,R_k}$. We may sometimes identify $\robots$ with the set of indices $\set{1,\ldots,k}$.
On a ring of size $n\geq k$, a \emph{(k,n)-configuration} of the robots 
(or simply a configuration if $n$ and $k$ are clear from the context) is
given by a vector $\config \in [0,n-1]^k$  associating
to each robot $R_i$ its position $\config(i)$ on the ring. 
We assume w.l.o.g. that positions
are numbered in the clockwise direction. 

A \emph{view} of a robot on this
configuration gives the distances between the robots, starting from
its neighbor, i.e. the robot positioned on the next occupied node (a distance equals to $0$ meaning that two robots are 
on the same node).  A \emph{view} $\mathbf{V}=\tuple{d_1,\dots, d_k}\in [0,n]^k$ is a 
$k$-tuple such that $\sum_{i=1}^k d_i=n$ and $d_1 \neq 0$. We let $\ViewSet{k}{n}$ be
the set of possible views for $k$ robots on a ring of size $n$.
Notice that all the robots
sharing the same position should have the same view. For instance, 
suppose that,
on a ring of size 10, 2 robots $R_1$, and $R_2$ 
are on the same position of the ring (say position 1), $R_3$ is at position 4, $R_4$ is  
at position 8, and $R_5$ is at position 9 (see Figure~\ref{fig:views3}). Then,
the view of $R_1$ and $R_2$ is $\tuple{3,4,1,2,0}$. It is interpreted by the fact that there is a robot at a distance $3$ (it is $R_3$), a robot a at distance  $3+4$ (it is $R_4$) and so on. 
We point out that all the robots at the same position share the same view. We as well suppose that in a view, the first distance is not $0$ (this is possible by putting $0$ at the `end' of the view instead). As a matter of fact in the example of Figure~\ref{fig:views3}, there is a robot at distance $3+4+1+2=10$ from $R_1$ (resp. $R_2$), which is $R_2$ (resp. $R_1$). The sum of the $d_i$ corresponds always to the size of the ring and here the fact that in the view of $R_1$ we have as last element $0$ signifies that there is a distance $0$ between the last robot (here $R_2$) and $R_1$. When looking in the opposite direction, their view becomes: $\tuple{2,1,4,3,0}$. 
Formally, for a view $\mathbf{V} = \tuple{d_1, \dots, d_k}\in [0,n]^k$, we
note $\revert{\mathbf{V}}= \tuple{d_j, \dots, d_1, d_k, \dots, d_{j-1}}$ the corresponding view
 when looking at the ring in the opposite direction, where $j$ is the greatest index such
 that $d_j\neq 0$. 
 
Given a configuration $\config \in [0,n-1]^k$
and a robot $R_i\in \robots$, the view of robot $R_i$ when looking in the clockwise direction,
 is given by $\ViewR{\config}{i}=\tuple{d_i(i_1),d_i(i_2)-d_i(i_1), \dots, n-d_i(i_{k-1})}$,
where, for all $j\neq i$, $d_i(j) \in [1,n]$ is such that
$(\config(i) + d_i(j)) \modulo n = \config(j)$ and $i_1,\dots, i_k$ are indexes pairwise
different such that
$0<d_i(i_1)\leq d_i(i_2)\leq \dots \leq d_i(i_{k-1})$.
When robot $R_i$ looks in the opposite direction, its view according to the
configuration $\config$ is $\ViewL{\config}{i}=\revert{\ViewR{\config}{i}}$.

\subsubsection{Protocols}
In our context, a protocol for networks of $k$
robots is given by a QFP formula respecting some specific
constraints.  
\begin{definition}[Protocol]\label{def:protocol}
A protocol is a QFP formula $\phi(x_1,\ldots,x_k)$ such that  for all
views $\mathbf{V}$ the following holds: if $\mathbf{V}
\models \phi$ and $\mathbf{V} \neq \revert{\mathbf{V}}$ then
$\revert{\mathbf{V}} \not\models\phi$
\end{definition}

A robot uses the protocol to know in which direction it should
move according to the following rules. As we have already stressed, all the robots that share the same position have the same view
of the ring.
Given a configuration $\config$ and a robot $R_i \in\robots$, if
$\ViewR{\config}{i} \models \phi$, then the robot $R_i$ moves in the clockwise
direction, if $\ViewL{\config}{i} \models \phi$ then it moves in the
opposite direction, if none of $\ViewR{\config}{i}$ and
$\ViewL{\config}{i}$ satisfies $\phi$ then the robot should not
move. The conditions expressed in Definition~\ref{def:protocol} imposes hence a direction when
$\ViewR{\config}{i} \neq \ViewL{\config}{i}$. In case
$\ViewR{\config}{i} = \ViewL{\config}{i}$, the robot is disoriented and it
can hence move in one direction or the other. For instance, consider the configuration
$\config$ pictured on Figure~\ref{fig:views1}. Here, $\ViewR{\config}{1}=\tuple{3,1,3}
=\ViewL{\config}{1}$.
Note that such a
semantics enforces that the behavior of a robot is not influenced by its
direction. In fact consider two symmetrical configurations $\config$
and $\config'$ such that $\ViewR{\config}{i}=\revert{\ViewR{\config'}{i}}$ for each robot $R_i$. 
If $\ViewR{\config}{i} \models \phi$ (resp. $\ViewL{\config}{i} \models
\phi$), then necessarily $\ViewL{\config'}{i} \models \phi$
(resp. $\ViewR{\config'}{i} \models \phi$), and the robot
in $\config'$ moves in the opposite direction than in $\config$
(and the symmetry of the two configurations is maintained).

We now formalize the way movement is decided. Given a protocol $\phi$ and a view $\mathbf{V}$, the moves of any
robot whose clockwise direction view is $\mathbf{V}$ are given by:
$$
\mov(\phi,V)= \left\{
	\begin{array}{ll}
		\set{+1}  & \mbox{if } \mathbf{V} \models \phi
                            \mbox{ and } \mathbf{V}
                      \neq \revert{\mathbf{V}} \\
		\set{-1} & \mbox{if } \revert{\mathbf{V}} \models \phi \mbox{ and } {\mathbf{V}}
                      \neq \revert{\mathbf{V}} \\
                \set{-1,+1} & \mbox{if } {\mathbf{V}} \models \phi \mbox{ and } {\mathbf{V}}
                      = \revert{\mathbf{V}} \\
            \set{0} & \mbox{otherwise}
	\end{array}
\right.
$$

Here $+1$ (resp. $-1$) stands for a movement of the robot in the clockwise (resp. anticlockwise) direction.

\subsection{Different possible semantics}

We now describe different transition relations between
configurations. Robots have a two-phase
behavior : (1) look at the ring and
(2) according to their view, compute and perform a movement. In this context, we
consider three different modes. In the \emph{semi-synchronous mode}, in one step,
some of the robots look at the ring and move. In the  \emph{synchronous
  mode}, in one step, all the robots look at the ring and
move. In the \emph{asynchronous mode}, in one step a single robot can
either choose to look at the ring, if the last thing it did was a
movement, or to move, if the last thing it did was to look at the ring. As a consequence, 
 its movement decision is a consequence of the view of
the ring it
has in its memory. In the remainder of the paper, we fix a protocol $\phi$ and we 
consider a set $\robots$
of $k$ robots.

\subsubsection{Semi-synchronous mode} We begin by providing the
semantics in the semi-synchronous case.
For this matter we define
the transition relation $\semtrans_\phi \subseteq [0,n-1]^k \times [0,n-1]^k$ (simply noted $\semtrans$
when $\phi$ is clear from the context) between configurations.
We have $\config \semtrans \config'$ if there exists a subset 
$I\subseteq \robots$ of robots such that, for all $i\in I$, $\config'(i)=(\config(i)+ m) \modulo n$, where 
$m\in \mov(\phi, \ViewR{\config}{i})$, 
and for all $i\in \robots\setminus I$, $\config'(i)=\config(i)$.

\subsubsection{Synchronous mode} 

The transition relation $\synchtrans_\phi
\subseteq  [0,n-1]^k \times [0,n-1]^k$ (simply noted $\synchtrans$ when $\phi$
is clear from the context) describing synchronous movements is
very similar to the semi-synchronous case, except that all
the robots have to move. Then $\config \synchtrans \config'$
if $\config'(i)=(\config(i)+m) \modulo n$ with
$m \in \mov(\phi, \ViewR{\config}{i})$
for all $i \in \robots$.

\subsubsection{Asynchronous mode} 

The definition of transition relation
for the asynchronous mode is a bit more involved, for two reasons: first, the move
of each robot does not depend on the current configuration, but on the last view of
the robot. 
Second, in one step a robot either look or move.
As a consequence, an \emph{asynchronous configuration} is a tuple 
$(\config,\internalState, \Views)$ where $\config \in [0,n-1]^k$ gives the current configuration,
$\internalState \in  \set{\lo,\mo}^k$ gives, for each robot, its internal state ($\lo$ stands for ready to look and 
 and $\mo$ stands for compute and move) and $\Views \in \ViewSet{k}{n}^k$ stores, for each robot, the view (in the clockwise direction) 
it had the last time it looked at the ring.

The
transition relation for asynchronous mode is hence defined by a binary
relation $\asynchtrans_\phi$ (or simply $\asynchtrans$) working on $[0,n-1]^k \times \set{\lo,\mo}^k \times \ViewSet{k}{n}^k$ and defined as follows:
$\tuple{\config, \internalState, \Views} \asynchtrans \tuple{\config',\internalState',\Views'}$
  iff there exist $R_i  \in \robots$  such that the following conditions
  are satisfied:
\begin{itemize}
\item for all $R_j\in\robots$ such that $j\neq i$, $\config'(j)=\config(j)$, $\internalState'(j)=\internalState(j)$ and $\Views'(j)=\Views(j)$,
\item if $\internalState(i)=\lo$ then $\internalState'(i)=\mo$, $\Views'(i)=\ViewR{\config}{i}$ and $\config'(i)=\config(i)$, i.e. if the robot
that has been scheduled was about to look, then the configuration of the robots won't change, and this robot updates its view of the ring
according to the current configuration and change its internal state,
\item if $\internalState(i)=\mo$ then $\internalState'(i)=\lo$, $\Views'(i)=\Views(i)$ and $\config'(i)=(\config(i) + m)\modulo n$, with $m\in\mov(\phi, \Views(i))$,
i.e. if the robot was about to move, then it changes its internal state and moves according to the protocol, and \emph{its last view of the ring}.
\end{itemize}

\subsubsection{Runs}
A semi-synchronous (resp. synchronous) $\phi$-run (or a run according to a protocol $\phi$) is a (finite or infinite) sequence of 
configurations $\rho=\config_0 \config_1\dots$ where, for all $0\leq i < |\rho|$, $\config_i\semtrans_\phi \config_{i+1}$ (resp. $\config_i\synchtrans_\phi \config_{i+1}$). Moreover, if $\rho=\config_0\cdots\config_n$ is finite, then there is no $\config$ such that
$\config_n\semtrans_\phi\config$ (respectively $\config_n\synchtrans_\phi\config$). \iflong \comment[ts]{Notion of initial asynchronous run or non initial asynchronous run would
be convenient in the proof of undecidability of \safety{\async}.}\fi An asynchronous $\phi$-run is a (finite or infinite) sequence of asynchronous
configurations $\rho = \tuple{\config_0,\internalState_0,\Views_0}\tuple{\config_1,\internalState_1,\Views_1}\cdots$ where, for all $0\leq i < |\rho|$, 
$\tuple{\config_i,\internalState_i,\Views_i}\asynchtrans_\phi \tuple{\config_{i+1},\internalState_{i+1},\Views_{i+1}}$ and such that $\internalState_0(i)=\lo$ for
all $i\in[1,k]$. Observe that the value of $\Views_0$ has no influence on the actual asynchronous run. 

We let $\post_\ssync(\phi,\config)=\{\config'\mid \config\semtrans_\phi\config'\}$, $\post_\sync(\phi,\config)=\{\config'\mid \config\synchtrans_\phi \config'\}$ 
and $\post_\async(\phi,\config)=\{\config'\mid \textrm{ there exist }\Views, \internalState', \Views'\textrm{ s.t. }
\tuple{\config, \internalState_0,\Views}\asynchtrans_\phi \tuple{\config',\internalState',\Views'}\}$, with $\internalState_0(i)=\lo$ for all $i\in [1,k]$.  Note that in the asynchronous case we impose all the robots to be ready to look.
We respectively write $\semtrans_\phi^*$, $\synchtrans_\phi^*$ and $\asynchtrans_\phi^*$ for the reflexive and transitive closure of the relations $\semtrans_\phi$,
$\synchtrans_\phi$ and $\asynchtrans_\phi$ and we define $\poststar{\ssync}(\phi,\config)$, 
$\poststar{\sync}(\phi,\config)$ and $\poststar{\async}(\phi,\config)$ by replacing in the definition $\post_{\ssync}(\phi,\config)$,$\post_{\sync}(\phi,\config)$ and $\post_{\async}(\phi,\config)$ the relations $\semtrans_\phi$, $\synchtrans_\phi$ and $\asynchtrans_\phi$ by their reflexive and transitive closure accordingly.
\iflong \comment[ts]{manque definition $\asynchtrans^n$ etc. utilisee dans preuve theoreme 1}\fi

We now come to our first result that shows that when the protocols have a special shape, the three semantics are identical.

\begin{definition}
A protocol $\phi$ is said to be \emph{uniquely-sequentializable} if, for all configuration $\config$, there is at most one robot $R_i\in\robots$ such that 
$\mov(\phi, \ViewR{\config}{i})\neq \{0\}$.
\end{definition}

When $\phi$ is uniquely-sequentializable at any moment at most one robot moves. Consequently, in that specific case, the three semantics are equivalent as stated by the following theorem.

\begin{theorem}\label{th:runs-synch-asynch}
If a protocol $\phi$ is uniquely-sequentializable, then for
all configuration $\config$,
$\poststar{\sync}(\phi,\config)=\poststar{\ssync}(\phi,\config)=\poststar{\async}(\phi,\config)$.
\end{theorem}

\iflong
\begin{proof}
Let $\phi$ be a protocol and $\config$ a configuration. From the definitions, it is obvious that $\poststar{\sync}(\phi,\config)\subseteq\poststar{\ssync}(\phi,\config)
\subseteq\poststar{\async}(\phi,\config)$. 

We show that, in case $\phi$ is uniquely-sequentializable, we also have $\poststar{\async}(\phi,\config)\subseteq 
\poststar{\sync}(\phi,\config)$. We first prove the following property $(P)$ of uniquely-sequentializable runs: for all uniquely-sequentializable protocol $\phi$, for all configuration $\config$, for all run
$\rho=\tuple{\config_0,\internalState_0,\Views_0}\tuple{\config_1,\internalState_1,\Views_1}\cdots\in \runs{\phi}{\async}$, for all $0\leq k<|\rho|$, for all robot $i$, if $\Views_k(i)\neq \ViewR{\config_k}{i}$ and $\internalState_k(i)=\mo$ then $\mov(\phi,\Views_k(i))=\{0\}$. 
To prove $(P)$, let $\rho\in \runs{\phi}{\async}$. We show $(P)$ by induction on $k$. For $k=0$, it is obvious since $\internalState_0(i)=\lo$ for all robot $i$.
Let $\rho= \tuple{\config_0,\internalState_0,\Views_0}\cdots\tuple{\config_k,\internalState_k,\Views_k}\tuple{\config_{k+1}, \internalState_{k+1},\Views_{k+1}}$
a prefix of a run in $\runs{\phi}{\async}$. Let $i$ be a robot such that $\internalState_{k+1}(i)=\mo$ and $\Views_{k+1}(i)\neq\ViewR{\config_{k+1}}{i}$. If 
$\internalState_k(i)=\lo$, then $\Views_{k+1}(i)=\ViewR{\config_k}{i}$ and $\config_{k}=\config_{k+1}$, which is impossible.
Hence, $\internalState_k(i)=\mo$ and $\Views_k(i)=\Views_{k+1}(i)$. Moreover, either $\Views_k(i)\neq\ViewR{\config_k}{i}$ and by induction hypothesis, 
$\mov(\phi,\Views_k(i))=\mov(\phi,\Views_{k+1}(i))=\{0\}$. Either $\Views_k(i)=\ViewR{\config_k}{i}\neq\ViewR{\config_{k+1}}{i}$. Hence there exists another
robot $j\neq i$ such that $\internalState_k(j)=\mo$ and $\internalState_{k+1}(j)=\lo$ and $\mov(\phi,\Views_k(j))\neq\{0\}$. Since $\phi$ is uniquely-sequentializable, 
$\mov(\phi,\Views_k(i))=\mov(\phi,\Views_{k+1}(i))=\{0\}$.

We show now by induction on $n$ that if $\tuple{\config,\internalState_0,\Views}\asynchtrans^n \tuple{\config',\internalState',\Views'}$ then 
$p'\in\poststar{\sync}(\phi)$. For $n=0$, it is obvious. Let now $n\in\nat$ and let $\config', \internalState',\Views', \config'', \internalState''$, and $\Views''$
such that $\tuple{\config,\internalState_0,\Views}\asynchtrans^n \tuple{\config',\internalState',\Views'}\asynchtrans\tuple{\config'',\internalState'',\Views''}$. By 
induction hypothesis, $\config'\in\poststar{\sync}(\config)$. Let $i$ be such that $\internalState''(i)\neq\internalState'(i)$. If $s'(i)=L$ then by definition, 
$\config''=\config'$ and $\config''\in\poststar{\sync}(\config)$. Otherwise, if $\ViewR{\config'}{i}\neq\Views'(i)$, then by $(P)$, $\mov(\phi,\Views'(i))=\{0\}$ and then
$\config'=\config''$. If $\ViewR{\config'}{i}=\Views'(i)$, either $\mov(\phi, \Views'(i))=\{0\}$ and $\config''=\config'$ or, $\mov(\phi,\Views'(i))\neq\{0\}$ and, since
$\phi$ is uniquely-sequentializable, for all $j\neq i$, $\mov(\phi,\ViewR{\config'}{j})=\{0\}$ and $\config'\synchtrans\config''$. Hence $\config''\in\poststar{\sync}(\config)$.
\end{proof}
\fi

\subsection{Problems under study}

In this work, we aim at verifying properties on protocols where we
assume that the number of robots is fixed (equals to $k>0$) but the
size of the rings is parameterized and satisfies a given
property. Note that when executing a protocol the size of the ring
never changes. For our problems, we consider a ring property that is given by a QFP formula  $\texttt{Ring}(y)$,  a set of bad configurations given by a QFP formula $\texttt{Bad}(x_1,\ldots,x_k)$ and a set of good configurations given by a QFP formula $\texttt{Goal}(x_1,\ldots,x_k)$. 
We then define two general 
problems to address the verification of such algorithms: the \safety{m} problem, and the \reachability{m} problem, with $m\in\{\ssync,\sync,
\async\}$.

The \safety{m} problem is to decide, given a protocol $\phi$ and two formulae $\texttt{Ring}$ and  $\texttt{Bad}$ whether
there exists a size $n\in\nat$ with $n\in \denote{\texttt{Ring}}$, and a $(k,n)$-configuration $\config$ with $\config\notin\denote{\texttt{Bad}}$, such that 
$\poststar{m}(\phi,\config)\cap\denote{\texttt{Bad}} \neq\emptyset$. 

The \reachability{m} problem is to decide given a protocol $\phi$ and two formulae $\texttt{Ring}$ and  $\texttt{Goal}$ whether
there exists a size $n\in\nat$ with $n\in \denote{\texttt{Ring}}$ and a $(k,n)$-configuration $\config$, such that $\poststar{m}(\phi,\config)\cap\denote{\texttt{Goal}}=\emptyset$. Note
that the two problems are not dual due to the quantifiers.

As an example, we can state in our context the \safety{m} problem that consists in checking  that a protocol $\phi$ working with three robots never leads to collision (i.e. to a configuration where two robots are on the same position on the ring) for rings of size strictly bigger than $6$. In that case we have $\texttt{Ring}:= y >6$ and $\texttt{Bad}:= x_1 = x_2 ~\vee~ x_2 = x_3 ~\vee~ x_1 = x_3$.

\section{Undecidability results}\label{sec:undecidability}
In this section, we present undecidability results for the two aforementioned problems. The proofs rely on the encoding of a deterministic
$k$-counter machine run. A deterministic $k$-counter machine consists of $k$ integer-valued registers (or counters) called
$c_1$, \dots, $c_k$, and a finite list of labelled instructions $L$. Each instruction is either of the form $\ell: c_i=c_i+1; \texttt{goto }\ell'$, 
or $\ell : \texttt{if }c_i>0 \texttt{ then } c_i=c_i-1;  \texttt{goto } \ell'; \texttt{else goto }\ell''$, where $i\in [1,k]$. We also assume the existence of a special instruction
$\ell_h: \texttt{halt}$. Configurations of a $k$-counter machine are elements of $L\times\nat^k$, giving the current instruction and the 
current values of the registers. The initial configuration is $(\ell_0,0,\dots,0)$, and the set of halting configurations is $\halt=\{\ell_h\}\times\nat^k$. Given a configuration $(\ell,n_1,\dots,n_k)$, the successor configuration 
$(\ell',n'_1,\dots,n'_k)$ is defined in the usual way
and we note $(\ell,n_1,\dots,n_k)\vdash (\ell',n'_1,\dots,n'_k)$.
A run of a $k$-counter machine is a (finite or infinite) sequence of configurations $(\ell_0,n^0_1,\dots,n^0_k), (\ell_1,n^1_1,\dots,n^1_k)\cdots $, where $(\ell_0,n^0_1,\dots,n^0_k)$ is
the initial configuration, and, for all $i\geq 0$, $(\ell_i, n^i_1, \dots,n^i_k)\vdash(\ell_{i+1}, n^{i+1}_1, \dots, n^{i+1}_k)$. The run is finite if and only if
it ends in
a halting configuration, i.e. in a configuration in \halt.

\begin{theorem}\label{th:undecidable-safe}
\safety{\async} is undecidable.
\end{theorem}

\iflong
\begin{proof}
\else
\begin{proofsketch}
\fi
The proof relies on a reduction from the halting problem of a deterministic two-counter machine $M$ to \safety{\async} with $k=42$ robots. It is likely
that an encoding using less robots might be used for the proof, but for the sake of clarity, we do not seek the smallest possible amount of robots. The halting
problem is to decide whether the run of a given deterministic two-counter machine is finite; this problem is undecidable~\cite{Minsky67}. The idea is to simulate the run of $M$ in a way that ensures that a 
collision occurs if and only if $M$ halts. Positions of robots on the
ring are used to encode values of
counters and the current instruction of the machine. The $k$-protocol makes sure that movements of the robots simulate correctly the run of $M$. Moreover,
one special robot moves only when the initial configuration is encoded, and another only when the final configuration is encoded. The collision is ensured in 
the following sequence of actions of the robots: when the initial configuration is encoded, the first robot computes its action but does not move immediately. When
the halting configuration is reached, the second robot computes its action and moves, then the first robot finally completes its move, entailing the collision. Note that if the ring is not big enough to simulate the counter values then the halting configuration is never reached and there is no collision.

Instead of describing configurations of the robots by applications giving positions of the robots on the ring, we use a sequence of letters $\free$ or $\robot$,
representing respectively a free node and a node occupied by a robot. When a letter $\textsf{A}\in\{\free,\robot\}$ is repeated $i$ times, we use
the notation $\textsf{A}^i$, when it is repeated an arbitrary number or times (including 0), we use $\textsf{A}^*$. To distinguish between the two representations
of the configurations, we use respectively the terms configurations or word-configurations. The correspondence between a configuration and a
word-configuration is obvious. A \emph{machine-like} (word-)configuration is a configuration of the form:

$\block{3}\free^*\robot\free^*\block{4}\free^*\robot\free^*\block{5}\free^*\robot\free^*\block{6}\free^*\robot\free^*\block{7}\free^*\robot\free^*\block{8}$$P_1$$P_2$$P_3$$P_4$$P_5$
$\robot
\free\robot$

where $\block{i}$ is a shorthand for $\free\robot^i\free$, and $P_1P_2\in\{\robot\free,\free\robot\}$ and exactly one $P_i=\robot$ for $i\in\{3,4,5\}$, 
$P_j=\free$ for $j\in \{3,4,5\}\setminus\{i\}$ (see Table~\ref{table:config} for a graphic representation of the section $P_1P_2P_3P_4P_5$ of the ring). 
Observe that the different blocks $\block{i}$ yield for every robot in the ring a distinct view. 
Hence, in the rest of the proof we abuse
notations and describe the protocol using different names for the different robots, according to their position in the ring, even if they are formally anonymous. We let
$\robots$ be the set of robots involved. A machine-like (word-)configuration 
$\block{3}$$\free^{n_1}$$\robot_{c_1}$$\free^*$$\block{4}$$\free^{n_2}$$\robot_{c_2}$$\free^*$$\block{5}$$\free^m$$\robot_c$$\free^n$$\block{6}$$\free^i$$\robot_\ell$$\free^{i'}$$\block{7}$$\free^p$$\robot_{\ell'}$$
\free^r$$\block{8}$
$\robot_{tt}$$\free$$\robot_t$$\free$$\free$$\robot_g$$\free$$\robot_d$ is said to be \emph{stable} because of the positions of robots $\robot_t$ and $\robot_{tt}$ (see
Table~\ref{table:config}). Moreover, it encodes the configuration $(\ell_i,n_1,n_2)$ of $M$ (due to the relative positions of robots $\robot_{c_1}$,
$\robot_{c_2}$ and $\robot_\ell$ respectively to $\block{3}$, $\block{4}$ and $\block{6}$). We say that a configuration $\config$ is machine-like, \stable, etc. if its corresponding
word-configuration is machine-like, \stable, etc. In the following, we distinguish configurations of the 2-counter machine, and configurations
of the robots, by calling them respectively $M$-configurations and $\phi$-(word)-configurations. For a stable and machine-like $\phi$-configuration $\config$, 
we let $M(\config)$ be the $M$-configuration encoded by $\config$.
We first present the part of the algorithm simulating the behavior of $M$. We call this algorithm $\phi'$.
Since the machine is deterministic, only one instruction is labelled by $\ell_i$, known by every robot. The simulation follows different steps, according to the positions
of the robots $\robot_{t}$ and $\robot_{tt}$, as pictured in Table~\ref{table:config}.

\begin{table}[htbp]
\caption{Different types of configurations}\label{table:config}
\centering
\begin{tabular}{|cccc|}
\hline
\stable configuration && $\robot_{tt} \free \robot_t \free\free$ & \begin{tikzpicture}[node distance =15em,>=stealth,-, scale=0.6, transform shape, thick]
 \draw (0,0) arc [x radius=2.5, y radius=-0.5, start angle=20, end angle=160]
\foreach \t in {0,0.25,0.75}
{ node [draw,circle,pos=\t, auto](\t) {} }
\foreach \t in {0.5,1}
{node [draw,circle,pos=\t, auto,fill=gray!100](\t) {}} ;

\end{tikzpicture}
\\
\moving configuration && $\free\robot_{tt}\robot_t \free\free$ &  \begin{tikzpicture}[node distance =15em,>=stealth,-, scale=0.6, transform shape, thick]
 \draw (0,0) arc [x radius=2.5, y radius=-0.5, start angle=20, end angle=160]
\foreach \t in {0,0.25,1}
{ node [draw,circle,pos=\t, auto](\t) {} }
\foreach \t in {0.5,0.75}
{node [draw,circle,pos=\t, auto,fill=gray!100](\t) {}} ;
\end{tikzpicture}
\\
\movingtwo configuration && $\free\robot_{tt} \free \robot_t\free$ & \begin{tikzpicture}[node distance =15em,>=stealth,-, scale=0.6, transform shape, thick]
 \draw (0,0) arc [x radius=2.5, y radius=-0.5, start angle=20, end angle=160]
\foreach \t in {0,0.5,1}
{ node [draw,circle,pos=\t, auto](\t) {} }
\foreach \t in {0.25,0.75}
{node [draw,circle,pos=\t, auto,fill=gray!100](\t) {}} ;
\end{tikzpicture}
\\
\movingthree configuration && $\free\robot_{tt} \free\free \robot_t$& \begin{tikzpicture}[node distance =15em,>=stealth,-, scale=0.6, transform shape, thick]
 \draw (0,0) arc [x radius=2.5, y radius=-0.5, start angle=20, end angle=160]
\foreach \t in {0.25,0.5,1}
{ node [draw,circle,pos=\t, auto](\t) {} }
\foreach \t in {0,0.75}
{node [draw,circle,pos=\t, auto,fill=gray!100](\t) {}} ;
\end{tikzpicture}
\\
\stabili configuration && $\robot_{tt} \free\free\free\robot_t$& \begin{tikzpicture}[node distance =15em,>=stealth,-, scale=0.6, transform shape, thick]
 \draw (0,0) arc [x radius=2.5, y radius=-0.5, start angle=20, end angle=160]
\foreach \t in {0.25,0.5,0.75}
{ node [draw,circle,pos=\t, auto](\t) {} }
\foreach \t in {0,1}
{node [draw,circle,pos=\t, auto,fill=gray!100](\t) {}} ;
\end{tikzpicture}
\\
\stabilitwo configuration && $\robot_{tt} \free\free \robot_t \free$& \begin{tikzpicture}[node distance =15em,>=stealth,-, scale=0.6, transform shape, thick]
 \draw (0,0) arc [x radius=2.5, y radius=-0.5, start angle=20, end angle=160]
\foreach \t in {0,0.5,0.75}
{ node [draw,circle,pos=\t, auto](\t) {} }
\foreach \t in {0.25,1}
{node [draw,circle,pos=\t, auto,fill=gray!100](\t) {}} ;
\end{tikzpicture}
\\
\hline
\end{tabular}
\end{table}%

We explain the algorithm $\phi'$ on the configuration $(\ell_i,n_1,n_2)$ with the transition $\ell_i:\texttt{if }c_1>0 \texttt{ then }
 c_1=c_1-1;  \texttt{goto } \ell_j; \texttt{else goto }\ell_{j'}$.
\begin{itemize}
\item When in a \stable configuration, robot $\robot_{tt}$ first moves to obtain a \moving configuration. 
\item In a \moving configuration, robot $\robot_c$ moves until it memorizes
the current value of $c_1$. More precisely, in a \moving configuration where $n_1\neq m$, robot $\robot_c$ moves : if $n_1>m$, and $n\neq 0$, $\robot_c$ moves towards $\block{6}$,
if $n_1<m$, it moves towards $\block{5}$, if $n_1>m$ and $n=0$, it does not move.
\item In a \moving configuration where $n_1=m$, $\robot_t$ moves to obtain a \movingtwo configuration.
\item In a \movingtwo configuration, if $n_1=m\neq 0$, then $\robot_{c_1}$ moves towards $\block{3}$, hence encoding the decrementation of $c_1$.
\item In a \movingtwo configuration, if $n_1=m=0$ or if $n_1\neq m$, (then the modification of $c_1$ is either impossible, or
already done), robot $\robot_{\ell'}$ moves until it memorizes the position of robot $\robot_\ell$: 
if $p<i$, and $r\neq 0$, $\robot_{\ell'}$ moves towards $\block{8}$; if $p>i$, $\robot_{\ell'}$ moves towards $\block{7}$. 
\item In a \movingtwo configuration, if $p=i$, then $\robot_t$ moves to obtain a \movingthree configuration.
\item In a \movingthree configuration, if $n_1=m=0$, and robot
  $\robot_{\ell'}$ encodes $\ell_i$ (i.e. $p=i$), then $c_1=0$ and
  robot $\robot_\ell$ has to move until it encodes $\ell_{j'}$. 
If on the other hand $n_1 < m$, then robot $\robot_\ell$
moves until it encodes $\ell_j$. More precisely, if $n_1=m=0$, and the position encoded by $\robot_\ell$ is smaller than $j'$ ($i<j'$),
and if $i'\neq 0$, then $\robot_\ell$ moves towards $\block{7}$. If $n_1=m=0$, and the position encoded by $\robot_\ell$ is greater than $j'$, $\robot_\ell$ moves towards $\block{6}$.
If $n_1\neq m$, then robot $\robot_\ell$ moves in order to reach a position where it encodes $\ell_{j}$ (towards $\block{6}$ if $i>j$, towards $\block{7}$ if $i<j$ and $i'\neq 0$). 
\item In a \movingthree configuration, if the position encoded by $\robot_{\ell'}$ is $\ell_i$, if $n_1=m=0$ and the position encoded by $\robot_\ell$ is $\ell_{j'}$, or if 
$n_1\neq m$,
 and the position encoded by $\robot_\ell$ is $\ell_j$, then the transition has been completely simulated : the counters have been updated and the next transition is
 stored. The robots then return to a \stable configuration: robot $\robot_{tt}$
moves to obtain a \stabili configuration.
\item In a \stabili configuration, robot $\robot_t$ moves to obtain a \stabilitwo configuration.
\item In a \stabilitwo configuration, robot $\robot_t$ moves to obtain a \stable configuration.
\end{itemize}

For other types of transitions, the robots move similarly. When in a \stable configuration encoding a configuration in \halt, 
no robot moves. 
We describe now the algorithm $\phi$ that simply adds to $\phi'$ the two following rules. 
Robot $\robot_g$ (respectively $\robot_d$) moves in the direction of $\robot_d$ (respectively in the direction of $\robot_g$)
if and only if the robots are in a \stable machine-like configuration, and the encoded configuration of the machine is $(\ell_0,0,0)$ (respectively is in \halt), 
(since the configuration is machine-like, the distance
between $\robot_g$ and $\robot_d$ is 2). 

On all configurations that are not machine-like, the algorithm makes sure that no robot move. This implies that once $\robot_g$
or $\robot_d$ has moved, no robot with a view up-to-date ever moves. One can easily be convinced that the algorithm can be
expressed by a QFP formula $\phi$. 

\iflong \fi

Let the formulae 
$
\texttt{Bad}(\config_1,\dots, \config_{42})=\bigvee_{\text{\scriptsize\tabular[t]{@{}l@{}}$i,j\in[1,42]$\\ $i\neq j$ \endtabular}} (\config_i=\config_j)
$ that is satisfied by all the 
configurations where two robots share the same position and $\texttt{Ring}(y)=y\geq 0$. We \iflong\else can \fi show that $M$ halts if and only if there exits a size $n\in\denote{\texttt{Ring}}$, a 
$(42,n)$-configuration $\config$ with $\config\notin\denote{\texttt{Bad}}$, such that 
$\poststar{\async}(\phi,\config)\cap\denote{\texttt{Bad}} \neq\emptyset$. 

\iflong
 For that matter, we use several claims about $\phi'$. First observe that $\phi'$ is uniquely-sequentializable.

 \begin{myclaim}\label{claim:machine-like}
Let $\config$ be a machine-like $\phi$-configuration. Then for all configuration $\config'\in\poststar{\async}(\phi',\config)=\poststar{\sync}(\phi',\config)$,
$\config'$ is machine-like. 
 \end{myclaim}

\begin{myclaim}\label{claim:if-successor-then-config}
Let $\tuple{\config,\internalState,\Views}$ be a \stable and machine-like asynchronous $\phi$-configuration corresponding to the following
word-configuration: $\block{3}$$\free^{n_1}$$\robot_{c_1}$$\free^{n'_1}$$\block{4}$$\free^{n_2}$$\robot_{c_2}$$\free^{n'_2}$$\block{5}$$\free^m$$\robot_c$$\free^n$$\block{6}$$
\free^i$$\robot_\ell$$\free^{i'}$$\block{7}$$\free^p$$\robot_{\ell'}$$
\free^r$$\block{8}$
$\robot_{tt}$$\free$$\robot_t$$\free$$\free$$\robot_g$$\free$$\robot_d$. Let $M(\config)=(\ell_i, n_1,n_2)$ a non-halting $M$-configuration. Assume that $\config$ has the following
properties.
$(i)$ if $n_1>m$ then $n\geq n_1 - m$ (if $\ell_i$ modifies $c_1$) or if $n_2>m$ then $n\geq n_2 - m$ (if $\ell_i$ modifies $c_2$),
$(ii)$ if $\ell_i$ increments $c_1$ (resp. $c_2$) then $n'_1>0$ (resp. $n'_2>0$), and 
$(iii)$ $i+i' = p+r = |L|$.
Then $\tuple{\config,\internalState,\Views}\asynchtrans^*_{\phi'} \tuple{\config',\internalState',\Views'}$ with $\config'$
\stable and machine-like, such that $M(\config)\vdash M(\config')$. 
\end{myclaim}

\begin{myclaim}\label{claim:if-config-then-successor}
Let $\tuple{\config,\internalState,\Views}$ and $\tuple{\config',\internalState',\Views'}$ be two asynchronous configurations, with $\config$ and $\config'$ \stable.
If there exists $k>0$ such that $\tuple{\config,\internalState,\Views} \asynchtrans^k_{\phi'} \tuple{\config',\internalState',\Views'}$ and for all $0<j<k$,
if $\tuple{\config,\internalState,\Views} \asynchtrans^j_{\phi'} \tuple{\config'',\internalState'',\Views''}$ we have $\config''$ not \stable, then 
$M(\config)\vdash M(\config')$.

\end{myclaim}

Assume that $M$ halts. There is then a finite bound $K\in\nat$ on the values of the counter during the run. We show hereafter an asynchronous 
$\phi$-run leading to a collision. Let $\tuple{\config_0,\internalState_0,\Views_0}$ be the initial configuration of the run, with $\config_0$ 
a $\phi$-configuration corresponding to the word-configuration 
\begin{equation*}
\block{3}\robot_{c_1}\free^{K}\block{4}\robot_{c_2}\free^{K}\block{5}\robot_c\free^K\block{6}\robot_\ell\free^{|L|}\block{7}\robot_{\ell'}\free^{|L|}\block{8}\robot_{tt}\free\robot_t\free\robot_g\free\robot_d,
\end{equation*} 
hence such that $M(\config_0)=(\ell_0,0,0)$, and $\internalState_0(i)=\lo$ 
for all $i\in\robots$. Then $\tuple{\config_0,\internalState_0,\Views_0}\asynchtrans \tuple{\config_1,\internalState_1,\Views_1}$ where
$\config_1=\config_0$, $\internalState_1(g)=\mo$, $\Views_1(g)=\ViewR{\config_0}{g}$, and $\internalState_1(i)=\internalState_0(i)$ and $\Views_1(i)=\Views_0(i)$
for all $i\in\robots\setminus\{g\}$. It is easy to check that,
at each step of a run starting from $\tuple{\config_0,\internalState_0,\Views_0}$, the conditions $(i)$, $(ii)$ and
$(iii)$ of Claim~\ref{claim:if-successor-then-config} are satisfied. Then, since $M$ halts, by applying iteratively Claim~\ref{claim:if-successor-then-config}, we have $\tuple{\config_1,\internalState_1,\Views_1}\asynchtrans^*_{\phi'} \tuple{\config_n,\internalState_n,\Views_n}$ with $M(\config_n)$ a halting configuration.
Now, the $\phi$-run continues with robot $d$ being scheduled to look, and then robots $g$
and $d$ moving, leading to a collision. Formally : 
$\tuple{\config_n,\internalState_n, \Views_n}\asynchtrans \tuple{\config,\internalState,\Views}\asynchtrans \tuple{\config',\internalState',\Views'}\asynchtrans
\tuple{\config'',\internalState'',\Views''}$, with $\config_n=\config$, $\internalState(d)=\mo$, and $\Views(d)=\ViewR{\config_n}{d}$, $\config'(d)=\config(d)-1$ since 
$\Views(d)$ memorizes the encoding of a halting configuration, with distance between $\robot_g$ and $\robot_d$ equal to 2. Last, $\config''(g)=\config(g)+1$
since $\Views(g)$ memorizes the encoding of the initial configuration, with distance between $\robot_g$ and $\robot_d$ equal to 1. Hence, $\config''(g)=\config''(d)$
and $\config''\in\poststar{as}(\phi,\config_0)\cap\denote{\texttt{Bad}}$.

Conversely assume there is an asynchronous $\phi$-run $\rho$ leading to a collision. By Claim~\ref{claim:machine-like}, if $\rho$ is in fact a $\phi'$ asynchronous
run, there is no collision (either $\rho$ starts in a non machine-like configuration and no robot moves, or it starts in a machine-like configuration and
it never reaches a collision, since machine-like configurations are collision-free by construction). Hence $\rho$ contains moves from $\robot_g$ and/or 
$\robot_d$. Assume $\robot_d$ moves in this run. Hence, there is a configuration $\tuple{\config,\internalState,\Views}$ in this run where $\robot_d$ has just
been scheduled to move and is hence such that $\internalState(\robot_d)=\mo$, $\Views(\robot_d)=\ViewR{\config}{\robot_d}$, with 
$\mov(\phi,\Views(\robot_d))\neq 0$. Hence, by definition of $\phi$, $\config$ is machine-like, \stable and $M(\config)$ is a halting $M$-configuration. Let $j\neq
\robot_g$ a robot such that $\internalState(j)=\mo$. If $\mov(\phi,V(j))\neq\{0\}$, by Proposition (P) from the proof of Theorem~\ref{th:runs-synch-asynch}
and since $\mov(\phi,V(j))=\mov(\phi',V(j))$, $V(j)=\ViewR{\config}{j}$, and it is impossible since $M(\config)$ is a halting configuration. Hence, from 
$\tuple{\config,\internalState,\Views}$ the only robots that can move are $\robot_d$ and $\robot_g$. If $\robot_g$ does not move, then all the following
configurations in the run is $\config'$ with $\config'(i)=\config(i)$ for all $i\neq\robot_d$ and $\config'(\robot_d)=\config(\robot_d)+1$. Hence, $\rho$ does not
lead to a collision.

Assume now that $\robot_g$ moves in $\rho$. Let $\rho'\cdot\tuple{\config_1,\internalState_1,\Views_1}\tuple{\config_2,\internalState_2,\Views_2}$ be a prefix of $\rho$ 
such that $\config_2(\robot_g)=\config_1(\robot_g)-1$ and $\config_2(i)=\config_1(i)$ for all $i\neq\robot_g$, and $\rho'\cdot\tuple{\config_1,\internalState_1,\Views_1}$ 
is a prefix of a $\phi'$-run. Then $\config_1$ is
a machine-like configuration and again by Proposition (P) of the proof of Theorem~\ref{th:runs-synch-asynch}, there is at most one robot $i\notin \{\robot_g,\robot_d\}$
such that $\mov(\phi',\Views_1(i))=\mov(\phi',\Views_2(i))\neq\{0\}$. Hence if $\robot_d$ does not move in $\rho$ then $\rho= \rho'\cdot\tuple{\config_1,\internalState_1,\Views_1}\tuple{\config_2,\internalState_2,\Views_2}\tuple{\config_3,\internalState_3,\Views_3}\cdot\rho''$ with $\config_3(i)=\config_2(i)$ for all
 the robots $i$ but possibly one, and then $\config_3$ being the only configuration appearing in $\rho''$. According to $\phi'$, $\config_3$ is collision free,
 so such a run could not yield a collision.
 
 Hence, we know that in $\rho$ both $\robot_g$ and $\robot_d$ moves. Moreover, from the above reasonings we deduce that in $\rho$, at some point $\robot_g$
 has been scheduled to look, then later on, $\robot_d$ has been scheduled to look, and just after, $\robot_g$ and $\robot_d$ have moved
 provoking a collision. Formally, $\rho$ is in the following form:
 $\rho=\rho'\cdot\tuple{\config_0,\internalState_0,\Views_0}\tuple{\config_1,\internalState_1,\Views_1}\cdot\rho''\cdot\tuple{\config_2,\internalState_2,\Views_2}
 \tuple{\config_3,\internalState_3,\Views_3}\tuple{\config_4,\internalState_4,\Views_4}\tuple{\config_5,\internalState_5,\Views_5}$ with $\rho',\rho''$ asynchronous $\phi'$-runs,
 $\internalState_0(\robot_g)=\lo$, $\internalState_1(\robot_g)=\mo$, and for all other robots $i$, $\internalState_0(i)=\internalState_1(i)$, and $\config_0=\config_1$, and 
 $\mov(\phi,\Views_1(\robot_g))\neq\{0\}$,
 $\config_2=\config_3$, $\internalState_2(\robot_d)=\lo$, $\internalState_3(\robot_d)=\mo$, $\mov(\phi,\Views_3(\robot_d))\neq\{0\}$, and either $\config_4(\robot_d)=\config_3(\robot_d)+1$, and for all $i\neq \robot_d$,
 $\config_3(i)=\config_4(i)$ and $\config_5(\robot_g)=\config_4(\robot_g)-1$ and for all $i\neq \robot_g$, $\config_4(i)=\config_5(i)$, or $\config_4(\robot_g)=\config_3(\robot_g)-1$,
  and for all $i\neq \robot_g$,
 $\config_3(i)=\config_4(i)$ and $\config_5(\robot_d)=\config_4(\robot_d)-1$ and for all $i\neq \robot_d$, $\config_4(i)=\config_5(i)$.
In both cases, we deduce that since $\mov(\phi,\Views_1(\robot_g))\neq\{0\}$, then $\config_0$ is a machine-like \stable configuration such that $M(\config_0)=C_0$,
and since $\mov(\phi,\Views_3(\robot_d))\neq\{0\}$, then $\config_2$ is a machine-like \stable configuration encoding a halting $M$-configuration $C_h$. Hence, from 
Claim~\ref{claim:if-config-then-successor}, since $\config_0\asynchtrans^*\config_2$ then $C_0\vdash\cdots \vdash C_h$ and $M$ halts.
\end{proof}
\else
 \qed
\end{proofsketch}
\fi

\begin{theorem}\label{th:universal-undecidable}
\reachability{m} is undecidable, for $m\in\{\ssync,\sync,\async\}$.
\end{theorem}

\iflong
\begin{proof}
\else
\begin{proofsketch}
\fi
The proof relies on a reduction from the repeated reachability problem of a deterministic three-counter zero-initializing bounded-strongly-cyclic machine $M$, 
which is undecidable \cite{Mayr03}. A counter machine is zero-initializing if from the initial instruction $\ell_0$ it first sets all the counters to 0. Moreover, an
infinite run is said to be \emph{space-bounded} if there is a value $K\in\nat$ such that all the values of all the counters stay below $K$ during the run. A counter
machine $M$ is
bounded-strongly-cyclic if every space-bounded infinite run starting from any configuration visits $\ell_0$ infinitely often.
The repeated reachability problem we consider is expressed as follows: given a 3-counter zero-initializing bounded-strongly-cyclic machine $M$, 
does there exist an infinite
space-bounded run of $M$?
A configuration of $M$ is encoded in the same fashion than in the proof of Theorem~\ref{th:undecidable-safe}, with 3 robots encoding the values
of the counters. 
\iflong A machine-like configuration in that case is of the form $\block{3}\free^{n_1}\robot_{c_1}\free^*\block{4}\free^{n_2}\robot_{c_2}\free^*\block{5}\free^{n_3}\robot_{c_3}\free^*\block{6}\free^m\robot_c\free^n\block{7}\free^i\robot_\ell\free^{i'}\block{8}\free^p\robot_{\ell'}
\free^r\block{9}\robot_{tt}\free\robot_t\free\free$.\fi
A transition of $M$ is simulated by the algorithm in the same way than above \emph{except that if a counter is to be
increased, the corresponding robot moves accordingly even if there is no room to do it, yielding a collision}. 
\iflong Hence, 
in any machine-like non \stable configuration, exactly one robot moves (hence the only finite runs are either ending in 
a halting configuration, or in a collision, which is not machine-like).
The algorithm that governs the robots in that case is called $\overline{\phi}$ and is a variant of $\phi'$. \comment[ts]{detail more the algorithm?}

Let $\texttt{Machine\_like}$, $\texttt{Halting}$ and $\texttt{Collision}$ be three QFP formulae, with 50 free variables, such that 
$\config\in\denote{\texttt{Machine\_like}}$ 
(respectively $\config\in\denote{\texttt{Halting}}$, $\config\in\denote{\texttt{Collision}}$) if
and only if $\config$ is machine-like (respectively iff $M(\config)$ -- the $M$-configuration encoded by $\config$ -- is a halting configuration, and iff $\config(i)=\config(j)$ for some $i,j\in\robots$).
We then let $\texttt{Goal}=\neg\texttt{Machine\_like}\vee \texttt{Halting}\vee \texttt{Collision}$ and $\texttt{Ring(y)}=y\geq 0$.

We \iflong now \else can \fi show that there is a size $n\in\denote{\texttt{Ring}}$, and a $(50,n)$-configuration $\config$ such that $\poststar{\sync}(\phi,\config)\cap\denote{\texttt{Goal}}=\emptyset$ if and only 
if there exists an infinite space-bounded run of $M$. From
Theorem~\ref{th:runs-synch-asynch}, this \iflong also provides \else provides \fi an undecidability proof for 
\reachability{\ssync} and
\reachability{\async}.
\fi
\iflong
We use the following claims, reminiscent of the claims used in the proof of Theorem~\ref{th:undecidable-safe}:

\begin{myclaim}\label{cl:pos-machine-like}
Let $\config$ be a machine-like configuration. Then, for all configuration $\config'\in\poststar{\sync}(\overline{\phi},\config)$, $\config'\in
\denote{\texttt{Machine\_like}\vee\texttt{Collision}}$.
\end{myclaim}

\begin{myclaim}\label{cl:successor-then-config}
Let $\config$ be a \stable and machine-like synchronous $\overline{\phi}$-configuration corresponding to the following
word-configuration: $$\block{3}\free^{n_1}\robot_{c_1}\free^{n'_1}\block{4}\free^{n_2}\robot_{c_2}\free^{n'_2}\block{5}\free^{n_3}\robot_{c_3}\free^{n'_3}\block{6}
\free^m\robot_c\free^n\block{7}\free^i\robot_\ell\free^{i'}\block{8}\free^p\robot_{\ell'}
\free^r\block{9}\robot_{tt}\free\robot_t\free\free.$$ Let $M(\config)=(\ell_i, n_1,n_2,n_3)$ be a non-halting $M$-configuration. Assume that $\config$ has the following
properties.
$(i)$ if $n_1>m$ (respectively $n-2>m$, $n-3>m$), then $n\geq n_1 - m$ (respectively $n\geq n_2-m$, $n\geq n_3-m$)(if $\ell_i$ modifies $c_1$ - respectively
$c_2$ or $c_3$), 
$(ii)$ if $\ell_i$ increments $c_1$ (resp. $c_2$ or $c_3$) then $n'_1>0$ (resp. $n'_2>0$, or $n'_3>0$), and 
$(iii)$ $i+i' = p+r = |L|$.
Then $\config\asynchtrans^*_{\overline{\phi}} \config'$ with $\config'$
\stable and machine-like, such that $M(\config)\vdash M(\config')$.
\end{myclaim}

\begin{myclaim}\label{cl:config-then-successor}
Let $\config, \config'$ be two \stable, machine-like configurations. If there exists some $k>0$ such that $\config\synchtrans_{\overline{\phi}}^k \config'$ and that for all $0<j<k$, 
if $\config\synchtrans_{\overline{\phi}}^j\config''$ then $\config''$ is not \stable, then $M(\config)\vdash M(\config')$.
\end{myclaim}

\begin{myclaim}\label{cl:stable}
Let $\config$ be a machine-like configuration, which is not \stable. Then, either $|\poststar{\sync}(\overline{\phi},\config)|$ is finite, or there exists 
$\config'\in\poststar{\sync}(\overline{\phi},\config)$ with $\config'$ \stable.
\end{myclaim}
If there is an infinite space-bounded run of $M$, we let $K\in\nat$ be the maximal values of all the counters during this run. Let $\config_0$ be the $\phi$-configuration
having the following word-representation:
\noindent$\block{3}\robot_{c_1}\free^{K}\block{4}\robot_{c_2}\free^{K}\block{5}\robot_{c_3}\free^K\block{6}\robot_c\free^K\block{7}\robot_\ell\free^{i'}\block{8}
\free^p\robot_{\ell'}\free^r\block{9}$\\
$\robot_{tt}\free\robot_t\free\robot_g\free\robot_d$. Hence $M(\config_0)$ is the initial configuration of $M$. It is easy to show that, for all 
$\config\in\poststar{\sync}(\overline{\phi}, \config_0)$, $\config$ satisfies conditions $(i)$, $(ii)$, and $(iii)$ of Claim~\ref{cl:successor-then-config}. Hence, by
applying iteratively Claim~\ref{cl:successor-then-config}, we can
build an infinite  $\overline{\phi}$-run $\rho=\config_0\config_1\cdots$
such that $\config_i\notin\denote{\texttt{Collision}\vee \texttt{Halting}}$ for all $i\geq 0$. Hence, by Claim~\ref{cl:pos-machine-like}, $\config_i\in\denote{\texttt{Machine\_like}}$
 for all $i\geq 0$, and $\config_i\notin\denote{\texttt{Goal}}$ for all $i\geq 0$. Moreover, for all machine-like configuration $\config$, for all $i\in\robots$, 
$\ViewR{\config}{i}\neq\ViewL{\config}{i}$, then it can have at most one successor, and $\rho$ is the only $\phi$-run starting from $\config_0$. Hence,
$\poststar{\sync}(\phi,\config_0)\cap\denote{\texttt{Goal}}=\emptyset$.

Conversely, assume that there is $n\in\nat$ and a $(k,n)$-$\phi$-configuration $\config_0$ such that 
$\poststar{\sync}(\overline{\phi},\config_0)\cap\denote{\texttt{Goal}}=\emptyset$. Hence, for all $\config\in\poststar{\sync}(\overline{\phi},\config_0)$,
$\config\in\denote{\texttt{Machine\_like}}$. Then, according to the definition
of the protocol, there is a unique synchronous $\overline{\phi}$-run starting from $\config_0$. Assume for the sake of contradiction that this run is finite. 
Let $\rho=\config_0\cdots\config_m$ be such a run. Then $\config_m\notin\denote{\texttt{Goal}}$ hence $M(\config_m)$ is not a halting configuration and $\config_m$
is collision-free. If it is not \stable, $\post_{\sync}(\overline{\phi},\config_m)\neq\emptyset$, from the definition of the protocol. If $\config_m$ is \stable, either 
$M(\config_m)=(\ell_h,n_1,n_2,n_3)$ for some $n_1,n_2,n_3 \in\nat$, but then $\config_m\in\denote{\texttt{Halting}}\subseteq\denote{\texttt{Goal}}$, which is not 
possible, or there exists a configuration $\config_{m+1}$ such that $\config_m\synchtrans\config_{m+1}$ and $\rho$ can be continued. Hence let $\rho=\config_0\config_1\cdots$ be the infinite synchronous $\overline{\phi}$-run
starting from $\config_0$. Assume that $\config_0$ is not \stable. Then, by Claim~\ref{cl:stable}, there exists $i\geq 0$, such that $\config_i$
is \stable. By definition of $\overline{\phi}$, $\config_{i+1}$ is not \stable, but by Claim~\ref{cl:stable} and Claim~\ref{cl:config-then-successor}, there exists $j>i+1$ such that
$\config_j$ is \stable and $M(\config_i)\vdash M(\config_j)$. By iterating this reasoning, we can in fact build an infinite run of $M$ starting in $M(\config_i)$. 
Let $K$ be the maximal number of positions between respectively $\block{3}$ and $\block{4}$, $\block{4}$ and $\block{5}$, $\block{5}$ and $\block{6}$
and $\block{6}$ and $\block{7}$ in $\config_0$. It is easy to see that this distance is an invariant of any $\overline{\phi}$-run. Hence, for any $k\geq 0$ such that $\config_k$
is \stable, $M(\config_k)=(\ell, n_1,n_2,n_3)$ with $n_i\leq K$ for $i\in\{1,2,3\}$, and the infinite run of $M$  is indeed space-bounded. Let $C_0\vdash C_1\cdots$
be such a run. Since $M$ is 
bounded-strongly-cyclic, there exists $i\geq 0$ such that $C_i=(\ell_0,n_1,n_2,n_3)$ with $n_i\in \nat$ for $i=\{1,2,3\}$, and since $M$ is zero-initializing, then 
there exists $j\geq i$ such that $C_j=(\ell_0,0,0,0)$. Hence,
$M$ has an infinite space-bounded run from $(\ell_0,0,0,0)$.
\fi

\iflong
\end{proof}
\else\qed
\end{proofsketch}
\fi

\section{Decidability results and case study}
In this section, we show that even if \safety{\async}, \reachability{\async}, \reachability{\ssync} and \reachability{\sync} are undecidable, the other cases \safety{\sync} and \safety{\ssync} can be reduced to the satisfiability problem for EP formulae, which is decidable and \textsc{NP}-complete \cite{Borosh&Treybig76}.
\subsection{Reducing safety to successor checking}
The first step towards decidability is to remark that to solve \safety{\sync} and \safety{\ssync} it is enough to look at the one-step successor.  Let $\phi$ be a protocol over $k$ robots and $\texttt{Ring}$ and  $\texttt{Bad}$ be respectively a ring property and a set of bad configurations. We have then the following lemma.

\begin{lemma}
\label{lem:poststar-to-post}
Let $n\in\nat$ such that $n\in \denote{\texttt{Ring}}$ and $m \in \set{\sync,\ssync}$. There exists a $(k,n)$-configuration $\config$ with $\config\notin\denote{\texttt{Bad}}$, such that 
$\poststar{m}(\phi,\config)\cap\denote{\texttt{Bad}} \neq\emptyset$ iff there exists a $(k,n)$-configuration $\config'$ with $\config'\notin\denote{\texttt{Bad}}$, such that 
$\post_{m}(\phi,\config')\cap\denote{\texttt{Bad}} \neq\emptyset$.
\end{lemma}

\iflong
\begin{proof}
The $\Leftarrow$ direction is obvious. For the $\Rightarrow$ direction, if there exists a synchronous or semi-synchronous  run $\rho=\config_0 \config_1\dots\config_n$  with $\config_0=\config$ and $\config_n$ being the first configuration in $\rho$ belonging to $\denote{\texttt{Bad}}$, then by taking $\config'=\config_{n-1}$ we have $\config'\notin\denote{\texttt{Bad}}$ and
$\post_{m}(\phi,\config')\cap\denote{\texttt{Bad}} \neq\emptyset$.
\end{proof}
\fi
This last result may seems strange at a first sight but it can easily be explained by the fact that robots protocols are most of the time designed to work without any assumption on the initial configuration, except that it is not a bad configuration.

\subsection{Encoding successor computation in Presburger}

We now describe various EP formulae to be used to express the computation of the successor configuration in synchronous and semi-synchronous mode.

First we show how to express the view of some robot  $R_i$ in a configuration $\config$, with the following formula:
$$
\begin{array}{c}
\configview{i}(y,p_1,\dots,p_k,d_1,\dots,d_k) := \\ 
\exists d'_1,\dots, d'_{k-1} \cdot \bigwedge_{j=1}^{k-2}d'_j\leq d'_{j+1}\wedge \\
\bigwedge_{j=1, j\neq i}^k (\bigvee_{\ell=1}^{k-1} p_j=(p_i+d'_\ell)\mod y) \wedge \\
 \bigwedge_{\ell=1}^k(\bigvee_{j=1, j\neq i}^k p_j=(p_i+d'_\ell)\mod y))\wedge\\
0<d'_1\wedge \bigwedge_{j=1}^{k-1}d'_j\leq y \wedge \\
 d_1=d'_1\wedge \bigwedge_{j=2}^{k-1}d_j=d'_j-d'_{j-1}\wedge d_k=y-d'_{k-1},
\end{array}
$$
Note that this formula only expresses in the syntax of Presburger arithmetic the definition of  $\ViewR{\config}{i}$ where the variable $y$ is used to store the length of the ring, $p_1,\ldots,p_k$ represent $\config$ and the variables $d_1,\ldots,d_k$ represent the view. \iflong In fact, we have the following statement. 
 \begin{lemma}\label{prop:configview}
 For all $i\in [1,k]$, we have  $n,\config,\mathbf{V} \models \configview{i}$ 
if and only if $\tuple{d_1,\dots, d_k} = \ViewR{\config}{i}$ on a ring of size $n$ .
\end{lemma}

\begin{proof}
Assume  $n,\config,\mathbf{V} \models \configview{i}$. Then, there exist $k-1$ variables, $d'_1,\dots, d'_{k-1}\in [1,n]$ such that 
$0<d'_1\leq d'_2\leq\dots\leq d'_{k-1}$. Moreover, there exists a bijection $f:[1,k-1]\rightarrow [1,k-1]$ such that, for all $j\neq i$,
$p_j = (p_i+d'_{f(j)}) \mod n$. Finally, $d_1=d'_1$ and for all $j\in [2,n-1]$, $d_j=d'_j-d'_{j-1}$ and $d_k=n-d'_{k-1}$. Hence, if we consider
the configuration $\config$ defined by $\config(j)=p_j$ for all $j\in[1,n]$, then $\tuple{d_1,\dots, d_k} = \ViewR{\config}{i}$.
Conversely, let $\config$ be a $(k,n)$-configuration and $\ViewR{\config}{i}=\tuple{d_1,\dots,d_k}$. Then, 
$n,\config,\mathbf{V} \models \configview{i}$. Indeed, by definition of the view, we let $d_i(j)\in [1,n]$ be such that 
$(\config(i) + d_i(j)) \modulo n = \config(j)$ for all $j\neq i$ and we let $i_1,\dots, i_k$ be a permutation of positions such that 
$0<d_i(i_1)\leq d_i(i_2)\leq\dots\leq d_i(i_{k-1})$. Then, for all $j\in [2,k-1]$, $d_j=d_i(i_j)-d_i(i_{j-1})$, $d_1=d_i(i_1)$ and $d_k=n-d_i(i_{k-1})$. By interpreting
the variables $d'_1,\dots, d'_k$ by respectively $d_i(i_1), \dots, d_i(i_k)$, it is easy to see that the formula is satisfied.
\end{proof}
\fi

We also use the formula $\viewsym(d_1,\dots, d_k, d'_1, \dots,
d'_k)$ that is useful to compute the symmetric of a view.
$$
\begin{array}{c}
\viewsym(d_1,\dots, d_k, d'_1, \dots, d'_k):= \\
\bigvee_{j=1}^k(\bigwedge_{\ell=j+1}^{k} (d_\ell=0 \wedge d'_\ell=0) \wedge \bigwedge_{\ell=1}^j d'_\ell=d_{j-\ell+1})
\end{array}
$$

\iflong\begin{lemma}\label{prop:viewsym}
For all $n\in\nat$, for all views $\mathbf{V},\mathbf{V}'\in [0,n]^k$,  we have $\mathbf{V},\mathbf{V'} \models \viewsym$ if and only if
 $\mathbf{V}'=\revert{\mathbf{V}}$.
\end{lemma}

\begin{proof}
Given $n\in\nat$, $d_1,\dots, d_k, d'_1,\dots, d'_k\in [0,n]$ such that $d_1\neq 0$ we have $\tuple{d'_1,\dots, d'_k}=\revert{\tuple{d_1,\dots,d_k}}$
if and only if there exists $1\leq j\leq k$ such that $d_\ell=0$ for all $j+1\leq \ell\leq k$ and $d'_1=d_j$, \dots, $d'_j=d_1$ and $d'\ell=0$ for all $j+1\leq \ell\leq k$
(by definition), if and only if $d_1,\ldots,d_k,d'_1,\ldots,d'_k\models \viewsym$.
\end{proof}
\fi

We are now ready to introduce the formula $\presburgmove^{\phi}_i(y,p_1,\dots, p_k, p')$, which is true if and only if, on a ring of size $n$ (represented by the variable $y$),
 the move of robot $R_i$
according to the protocol $\phi$ from the configuration $\config$ yields to the new position $p'$. Here the variables $p_1,\ldots,p_k$ characterizes  $\config$.
. 

$$
\begin{array}{c}
\presburgmove^\phi_i(y,p_1,\dots, p_k,p'):= \\
\exists d_1,\dots,d_k, d'_1, \cdots, d'_k \cdot \\
\configview{i}(y,p_1,\dots, p_k, d_1,\dots, d_k)\wedge\\
 \viewsym(d_1,\dots, d_k,d'_1,\dots, d'k)\wedge \\
 \Bigl[\Bigl(\phi(d_1,\dots,d_k)\wedge \bigl((p_i<y-1\wedge p'=p_i+1)\\
\vee (p_i=y-1\wedge p'=0)\bigr )\Bigr) \vee \\
\Bigl ( \phi(d'_1,\dots, d'_k) \wedge \bigl((p_i>0 \wedge p'=p_i-1) \\
\vee(p_i=0 \wedge p'=y-1) \bigr)\Bigr)\\
\vee \Bigl(\neg\phi(d_1,\dots, d_k)\wedge\neg\phi(d'_1,\dots, d'_k)\wedge(p'=p_i)\Bigr)\Bigr]
\end{array}
$$
\iflong
\begin{lemma}\label{prop:presburgmove}
For all $n\in\nat$ and a $(k,n)$-configurations $\config$ and $p' \in [0,n-1]$, we have $n,\config,p' \models \presburgmove^\phi_i$ if and only if $p'=(\config(i)+m)\modulo n$ with $m\in\mov(\phi,\ViewR{\config}{i})$.
\end{lemma}

\begin{proof}
We have $n,\config,\config' \models \presburgmove^\phi_i$ if and only if there exist $d_1,\dots, d_k, d'_1, \dots, d'_k\in [0,n]$ such that
$\tuple{d_1,\dots, d_k} = \ViewR{\config}{i}$ (by Lemma~\ref{prop:configview}) and $\tuple{d'_1,\dots,d'_k}=\revert{\tuple{d_1,\dots,d_k}}=\ViewL{\config}{i}$
(by Lemma~\ref{prop:viewsym}) and either $(i)$ $\ViewR{\config}{i}\models\phi$ and $p'=(p_i+1) \modulo n$, or $(ii)$ $\ViewL{\config}{i}\models\phi$ and
$p'=(p_i-1)\modulo n$, or $(iii)$ $\ViewR{\config}{i}\not\models\phi$, $\ViewL{\config}{i}\not\models\phi$ and $p'=p_i$, if and only if 
either $(i)$ $1\in \mov(\phi, \ViewR{\config}{i})$ and $p'=(p_i+1)\modulo n$ or $(ii)$ $-1\in\mov(\phi,\ViewR{\config}{i})$ and $p'=(p_i-1)\modulo n$ or 
$(iii)$ $\mov(\phi,\ViewR{\config}{i})=\{0\}$ and $p'=p_i$ if and only if $p'=(p_i+m)\modulo n$ with $m\in\mov(\phi,\ViewR{\config}{i})$.
\qed
\end{proof}
\fi
Now, given two $(k,n)$-configurations $\config$ and $\config'$, and a $k$-protocol $\phi$, it is easy to express the fact
that $\config'$ is a successor configuration of $\config$ according to $\phi$ in a semi-synchronous run (resp. synchronous run); for this we define the two formulae $\semisyncpost_\phi(y,p_1,\dots, p_k, p'_1,\dots, p'_k)$  and $\syncpost_\phi(y,p_1\dots, p_k, p'_1,\dots, p'_k)$)  as follows:

$$
\begin{array}{c}
\semisyncpost_\phi(y,p_1,\dots, p_k, p'_1, \dots, p'_k):= \\
\bigvee_{i=1}^k \bigr( \presburgmove^\phi_i(y,p_1, \dots, p_k, p'_i) \wedge\\
\bigwedge_{j=1, j\neq i}^k ((p'_j=p_j)\vee \presburgmove^\phi_j(y,p_1, \dots, p_k, p'_j))\bigr)\\
\textcolor{white}{aa}\\
\syncpost_\phi(y,p_1,\dots, p_k, p'_1, \dots, p'_k):=\\\bigwedge_{i=1}^k\presburgmove^\phi_i(y,p_1,\dots,p_k,p'_i)
\end{array}
$$

\begin{lemma}
\label{lem:post-to-Presburger}
For all $n \in \nat$ and all $(k,n)$-configurations $\config$ and $\config'$, we have:
\begin{enumerate} 
\item $\config\semtrans \config'$ if and only if
$n,\config,\config' \models \semisyncpost_\phi$,
\item $\config\synchtrans \config'$
if and only if $n,\config,\config' \models \syncpost_\phi$.
\end{enumerate}
\end{lemma}

\subsection{Results}

Now since to solve $\safety{\ssync}$ and $\safety{\sync}$, we only need to look at the successor in one step, as stated by Lemma \ref{lem:poststar-to-post}, and thanks to the formulae $ \semisyncpost_\phi$ and $ \syncpost_\phi$ and their properties expressed by Lemma  \ref{lem:post-to-Presburger}, we deduce that these two problems can be expressed in Presburger arithmetic.

\begin{theorem}
\safety{\sync} and \safety{\ssync} are decidable and in \textsc{NP}.
\end{theorem}

\begin{proof}
We consider a ring property  $\texttt{Ring}(y)$, a protocol $\phi$ for $k$ robots (which is a QFP formula) and a set of bad configurations given by a QFP formula $\texttt{Bad}(x_1,\ldots,x_k)$. We know that there exists a size $n\in\nat$ with $n\in \denote{\texttt{Ring}}$, and a $(k,n)$-configuration $\config$ with $\config\notin\denote{\texttt{Bad}}$, such that 
$\poststar{\sync}(\phi,\config)\cap\denote{\texttt{Bad}} \neq\emptyset$ if and only if  there exists a $(k,n)$-configuration $\config'$ with $\config'\notin\denote{\texttt{Bad}}$, such that 
$\post_{m}(\phi,\config')\cap\denote{\texttt{Bad}} \neq\emptyset$. By Lemma \ref{lem:post-to-Presburger}, this latter property is true if and only if the following formula is satisfiable:

$$
\begin{array}{c}
\syncpost_\phi(y,p_1,\dots, p_k, p'_1, \dots, p'_k) \wedge \\
\texttt{Ring}(y) \wedge \neg \texttt{Bad}(p_1,\ldots,p_k) \wedge \\
\texttt{Bad}(p'_1,\ldots,p'_k)
\end{array}
$$
For the semi-synchronous case, we replace the formula $\syncpost_\phi$ by $\semisyncpost_\phi$. The \textsc{NP} upper bound is obtained by the fact that the built formula is an EP formula.
\end{proof}

When the protocol $\phi$ is uniquely-sequentializable, i.e. when in each configuration at most one robot make the decision to move then Theorem \ref{th:runs-synch-asynch} leads us to the following result.

\begin{corollary}
\label{cor:unique-sequentializable}
When the protocol $\phi$ is uniquely-sequentializable, \safety{\async} is decidable.
\end{corollary}

\subsection{Expressing other interesting properties}

Not only the method consisting in expressing the successor
computation in Presburger arithmetic allows us to obtain the
decidability for \safety{\sync} and \safety{\ssync}, but they also
allow us to express other interesting properties. For instance, we can compute the successor configuration in asynchronous mode for a protocol $\phi$ working over $k$ robots thanks to the formula $\asyncpost_\phi(y,p_1,\ldots,p_k,s_1,\ldots,s_k,\linebreak[0]v_1\ldots,v_k, p'_1,\ldots,p'_k,s'_1,\ldots,s'_k,v'_1,\ldots,v'_k)$, which is given by:
$$
\begin{array}{c}
\asyncpost_\phi:= \exists d_1,\dots, d_k\cdot \\
\bigvee_{i=1}^k\Bigr( \bigwedge_{j\neq i} (p'_j=p_j\wedge s'_j=s_j\wedge v'_j=v_j) \wedge \\
 s'_i=1-s_i\wedge \bigl((s_i=0\wedge v'_i=\tuple{d_1,\dots, d_k}\wedge \\
\configview{i}(n,p_1, \dots, p_k, d_1, \dots, d_k) \wedge p'_i=p_i) \vee\\
 (s_i=1\wedge v'_i=v_i\wedge \presburgmove_i^\phi(n,p_1,\dots, p_k, p'_i)\bigr)\Bigr)
\end{array}
$$

To prove the correctness of this formula for an asynchronous configuration $(\config,\internalState, \Views)$ with $k$ robots we make the analogy between the flags $\lo$ and $\mo$ and the naturals $0$ and $1$, which means that in the definition of the vector $\internalState \in 
  \set{\lo,\mo}^k$, we encode $\lo$ by $0$ and $\mo$ by $1$ and we then apply the definition of $\asynctrans$.

\iflong
\begin{lemma}
For all $n \in \nat$ and all $(k,n)$ asynchronous configurations$\tuple{\config,\internalState,\Views}$ and $\tuple{\config',\internalState',\Views'}$, we have $\tuple{\config,\internalState,\Views}\asynchtrans \tuple{\config',\internalState',\Views'}$ if and only if
$n,\config,\internalState,\Views,\config',\internalState',\Views' \models \asyncpost_\phi$.
\end{lemma}
\fi

One can also express the fact that one configuration is a predecessor of the other in a straightforward way.

It is as well possible to check whether a protocol $\phi$ over $k$
robots fits into the hypothesis of Corollary
\ref{cor:unique-sequentializable}, i.e. whether it is
uniquely-sequentializable. We define the formula
$\texttt{UniqSeq}_\phi$ that is satisfiable if and only if $\phi$ is uniquely-sequentializable.
$$
\begin{array}{c}
\texttt{UniqSeq}_\phi :=\neg \exists y.p_1,\dots,p_k,p'_1,\dots,p'_k\cdot\\
\bigvee_{i\neq j, 1\leq i,j\leq k}
(\presburgmove^\phi_i(n,p_1,\dots,p_k,p'_i)\wedge  \presburgmove^\phi_j(n,p_1,\dots,p_k,p'_j)\\
\wedge p'_i\neq p_i\wedge p'_j\neq p_j.
\end{array}
$$

Hence we deduce the following statement.
\begin{theorem}\label{th:unique_decidable}
Checking whether a protocol $\phi$ is uniquely-sequentializable \iflong or pending-bounded \fi is decidable. 
\end{theorem}

\subsection{Applications}

We have considered the \emph{exclusive perpetual exploration} algorithms proposed by Blin \emph{et al.}~\cite{BlinMPT10}, and generated the formulae to check that no collision are encountered for different cases. We have used the SMT solver Z3 \cite{z3-solver} to verify whether the generated formulae were satisfiable or not. We have been able to prove that, in the synchronous case, the algorithm using a minimum of $3$ robots was safe for any ring of size greater than $10$  and changing a rule of the algorithm has allowed us to prove that we could effectively detect bugs in the Algorithm. In fact, in this buggy case, the SMT solver provides a configuration leading to a collision after one step. We have then looked for absence of collision for the algorithms using a maximum number of robots, always in the synchronous case. Here, the verification was not parametric as the size of the ring is fixed and depends on the number of robots (it is exactly $5$ plus the number of robots). The objective was to see whether our approach could be applied for a large number of robots. In that case, we have been able to prove that the algorithm for $6$ robots was safe, but we found some bugs for $5$, $7$, $8$, $9$, $10$ and $12$ robots. It was stated in~\cite{BlinMPT10} that the algorithm was not working for $5$ robots however the other cases are new bugs. Note however that for $11$ robots, the SMT solver Z3 was taking more than 10 minutes and we did not let him finish its computation. We observe that  when there is a bug, the SMT solver goes quite fast to generate a bad configuration but it takes much more time when the algorithm is correct, as for instance with $6$ robots. The files containing the SMT formulae are all available on the webpage \cite{experiments} in the SMTLIB format.

\section{Conclusion}

We have addressed two main problems concerning formal verification of protocols of mobile robots, and answered the open questions regarding decidability
of the verification of such protocols, when the size of the ring is given as a parameter of the problem. 
Note that in such algorithms, robots can start in any position on the ring. Simple modifications of the proofs in this paper allow to obtain undecidability of both
the reachability and the safety problem in any semantics, when the starting configuration of the robots is given. Hence we give a precise view of what can
be achieved in the automated verification of protocols for robots in the parameterized setting, and provide a means of partially verifying them. 
Of course, to fully demonstrate the correctness of a tentative protocol, more properties are required (like, all nodes are visited infinitely often) that are not handled
 with our approach. Nevertheless, as intermediate lemmas (for arbitrary $n$) are verified, the whole process of proof writing is both eased and strengthened.

An application of Corollary~\ref{cor:unique-sequentializable} and Theorem~\ref{th:unique_decidable}  deals with robot program synthesis as depicted in the approach of Bonnet \emph{et al.}~\cite{bonnet14wssr}. To simplify computations and save memory when synthesizing mobile robot protocols, their algorithm only generates uniquely-sequentializable protocols (for a given $k$ and $n$). Now, given a protocol description for a given $n$, it becomes possible to check whether this protocol remains uniquely-sequentializable for any $n$. Afterwards, regular safety properties can be devised for this tentative protocol, for all models of computation (that is, FSYNC, SSYNC, and ASYNC). Protocol design is then driven by the availability of a uniquely-serializable solution, a serious asset for writing handwritten proofs (for the properties that cannot be automated). 

Last, we would like to mention possible applications of our approach for problems whose core properties seem related to reachability only. One such problem is exploration with stop~\cite{berard16dc}: robots have to explore and visit every node in a network, then stop moving forever, assuming that all robots initial positions are distinct. All of the approaches published for this problem make use of \emph{towers}, that is, locations that are occupied by at least two robots, in order to distinguish the various phases of the exploration process (initially, as all occupied nodes are distinct, there are no towers). Our approach still makes it possible to check if the number of created towers remains acceptable (that is below some constant, typically 2 per block of robots that are equally spaced) from any given configuration in the algorithm, for any ring size $n$. As before, such automatically obtained lemmas are very useful when writing the full correctness proof.

\iflong
\begin{definition}
A protocol $\phi$ is \emph{pending-bounded} (? name ok?) if for all configuration $\config$ such that $|\actif{\config}|>1$\comment[ts]{rajouter la def : act($\config$) = tous
les robots tq $\mov(\phi,\ViewR{\config}{i})\neq{0}$}, for all configuration $\config'\in\post_{\ssync}(\config)\setminus\post_\sync(\config)$, let $J=\{i\in\robots\mid\config'(i)\neq
\config(i)\}$, then $\actif{\config'}\subseteq \actif{\config}\setminus J$ and for all $i\in\actif{\config'}$, $\mov(\phi,\ViewR{\config'}{i})=\mov(\phi,\ViewR{\config}{i})$.
\end{definition}

\begin{theorem}
When the protocol $\phi$ is pending-bounded, \safety{\async} is decidable.\comment[ts]{c'est mon intuition mais je n'ai pas encore de preuve qui me convainque. Si ca se trouve c'est faux}
\end{theorem}

\begin{proof}
Intuitively, when a protocol is pending-bounded, pending moves (i.e. moves that a robot has computed but not done yet) are not held for too long, and it is possible to bound
the length of a run leading to a bad configuration.

We first introduce the notion of pseudo-run. A sequence $\rho=\tuple{\config_0,\internalState_0,\Views_0}\tuple{\config_1,\internalState_1,\Views_1}\dots$
is a \emph{pseudo-run}, if there is an asyncrhonous run 
$\rho'=\tuple{\config'_0,\internalState'_0,\Views'_0}\tuple{\config'_1,\internalState'_1,\Views'_1}\dots$ such that 
for all $k\geq 0$, $\config'_k=\config_k$, $\internalState'_k=\internalState_k$, and, for all robot $i$, 
$\mov(\phi,\ViewR{\config'_k}{i})=\mov(\phi,\ViewR{\config_k}{i})$.

Let $\rho=\tuple{\config_0, \internalState_0,\Views_0}\tuple{\config_1,\internalState_1,\Views_1}\dots\tuple{\config_n,\internalState_n,\Views_n}$ a finite 
asynchronous run with pending moves, and assume without loss of generality that whenever a robot $i$ is scheduled to look, the protocol makes it actually move,
i.e. for all $0\leq k\leq n$, if $\internalState_k(i)=\mo$, then $\mov(\phi,\Views_k(i))\neq\{0\}$. Let $k$ be the first position such that there exists a robot $i_1$ such that $\internalState_k(i_1)=\lo$, 
$\mov(\phi,\ViewR{\config_k}{i_1})\neq\{0\}$, $\internalState_{k+1}(i_1)=\mo$ and there exists a robot $j\neq i_1$ such that $\internalState_k(j)=\mo$,
and $\Views_k(j)\neq\ViewR{\config_k}{j}$. Let $\tuple{\config_\ell,\internalState_\ell,\Views_\ell}$ be the moment where robot $j$ has taken its view (i.e.
$\internalState_\ell(j)=\lo$ and $\internalState_{\ell+1}=\mo$) and let $i_2\neq i_1$ be a robot such that $\config_\ell(i_2)\neq\config_k(i_2)$, and $\ell_2$ be the
position where it has moved: $\tuple{\config_{\ell_2},\internalState_{\ell_{2}},\Views_{\ell_2}}\tuple{\config_{\ell_2+1},\internalState_{\ell_{2}+1},\Views_{\ell_2+1}}$, 
and $\internalState_{\ell_2}(i_2)=\mo$ and $\internalState_{\ell_2+1}(i_2)=\lo$. If 
$\Views_{\ell_2}(i_2)\neq\ViewR{\config_\ell}{i_2}$, then the moment when $i_2$ took its view is the first position where there is a pending move. But we have 
assumed that $k$ was the first one. Hence, $\config_k\in\post_\ssync(\config_\ell)$. Moreover, since $\mov(\phi,\ViewR{\config_k}{j}\neq\{0\}$ but $\config_\ell(j)=
\config_k(j)$, $\config_k\notin\post_\sync(\config_\ell)$.  Let $J=\{i\mid \config_k(i)\neq\config_\ell(i)\}$. Since $\phi$ is pending-bounded, $\actif{\config_k}
\subseteq\actif{\config_\ell}\setminus J$. Hence, $i_1\in\actif{\config_\ell}\setminus J$ and $\mov(\phi,\ViewR{\config_\ell}{i_1})=\mov(\phi,\ViewR{\config_k}{i_1})$.

We define the pseudo-run $\rho'=\tuple{\config_0,\internalState_0,\Views_0}\dots\tuple{\config_\ell,\internalState_\ell,\Views_\ell}
\tuple{\config'_\ell,\internalState'_\ell,\Views'_\ell}\tuple{\config_{\ell+1},\internalState'_{\ell+1},\Views'_{\ell+1}}\dots\tuple{\config_k,\internalState'_k,\Views'_k}$
where, for all $0\leq m\leq k$, $\internalState'_m(i)=\begin{cases}\mo & \textrm{ if $i=i_1$}\\
\internalState_m(i) & \textrm{ otherwise}\end{cases}$ 
and $\Views'_m(i)=\begin{cases}\ViewR{\config'_\ell}{i} & \textrm{ if $i=i_1$}\\
\Views_m(i) & \textrm{ otherwise.}\end{cases}$ 

\comment[ts]{a finir}
\end{proof}
\fi



\iflong
\else
\newpage

\appendix

\section{Proof of Theorem~\ref{th:runs-synch-asynch}}

Let $\phi$ be a protocol and $\config$ a configuration. From the definitions, it is obvious that $\poststar{\sync}(\phi,\config)\subseteq\poststar{\ssync}(\phi,\config)
\subseteq\poststar{\async}(\phi,\config)$. 

We show that, in case $\phi$ is uniquely-sequentializable, we also have $\poststar{\async}(\phi,\config)\subseteq 
\poststar{\sync}(\phi,\config)$. We first prove the following property $(P)$ of uniquely-sequentializable runs: for all uniquely-sequentializable protocol $\phi$, for all configuration $\config_o$, for all runs
$\rho=\tuple{\config_0,\internalState_0,\Views_0}\tuple{\config_1,\internalState_1,\Views_1}\cdots\in \runs{\phi}{\async}$, for all $0\leq k<|\rho|$, for all robot $i$, if $\Views_k(i)\neq \ViewR{\config_k}{i}$ and $\internalState_k(i)=\mo$ then $\mov(\phi,\Views_k(i))=\{0\}$. 
To prove $(P)$, let $\rho\in \runs{\phi}{\async}$. We show $(P)$ by induction on $k$. For $k=0$, it is obvious since $\internalState_0(i)=\lo$ for all robot $i$.
Let $\rho= \tuple{\config_0,\internalState_0,\Views_0}\cdots\tuple{\config_k,\internalState_k,\Views_k}\tuple{\config_{k+1}, \internalState_{k+1},\Views_{k+1}}$
a prefix of a run in $\runs{\phi}{\async}$. Let $i$ be a robot such that $\internalState_{k+1}(i)=\mo$ and $\Views_{k+1}(i)\neq\ViewR{\config_{k+1}}{i}$. If 
$\internalState_k(i)=\lo$, then $\Views_{k+1}(i)=\ViewR{\config_k}{i}$ and $\config_{k}=\config_{k+1}$, which is impossible.
Hence, $\internalState_k(i)=\mo$ and $\Views_k(i)=\Views_{k+1}(i)$. Moreover, either $\Views_k(i)\neq\ViewR{\config_k}{i}$ and by induction hypothesis, 
$\mov(\phi,\Views_k(i))=\mov(\phi,\Views_{k+1}(i))=\{0\}$. Either $\Views_k(i)=\ViewR{\config_k}{i}\neq\ViewR{\config_{k+1}}{i}$. Hence there exists another
robot $j\neq i$ such that $\internalState_k(j)=\mo$ and $\internalState_{k+1}(j)=\lo$ and $\mov(\phi,\Views_k(j))\neq\{0\}$. Since $\phi$ is uniquely-sequentializable, 
$\mov(\phi,\Views_k(i))=\mov(\phi,\Views_{k+1}(i))=\{0\}$.

We show now by induction on $n$ that if $\tuple{\config,\internalState_0,\Views}\asynchtrans^n \tuple{\config',\internalState',\Views'}$ then 
$p'\in\poststar{\sync}(\phi)$. For $n=0$, it is obvious. Let now $n\in\nat$ and let $\config', \internalState',\Views', \config'', \internalState''$, and $\Views''$
such that $\tuple{\config,\internalState_0,\Views}\asynchtrans^n \tuple{\config',\internalState',\Views'}\asynchtrans\tuple{\config'',\internalState'',\Views''}$. By 
induction hypothesis, $\config'\in\poststar{\sync}(\config)$. Let $i$ be such that $\internalState''(i)\neq\internalState'(i)$. If $s'(i)=L$ then by definition, 
$\config''=\config'$ and $\config''\in\poststar{\sync}(\config)$. Otherwise, if $\ViewR{\config'}{i}\neq\Views'(i)$, then by $(P)$, $\mov(\phi,\Views'(i))=\{0\}$ and then
$\config'=\config''$. If $\ViewR{\config'}{i}=\Views'(i)$, either $\mov(\phi, \Views'(i))=\{0\}$ and $\config''=\config'$ or, $\mov(\phi,\Views'(i))\neq\{0\}$ and, since
$\phi$ is uniquely-sequentializable, for all $j\neq i$, $\mov(\phi,\ViewR{\config'}{j})=\{0\}$ and $\config'\synchtrans\config''$. Hence $\config''\in\poststar{\sync}(\config)$.


\section{Proof of Theorem~\ref{th:undecidable-safe}}
We show that $M$ halts if and only if there exits a size $n\in\nat$, a 
$(42,n)$-configuration $\config$ with $\config\notin\denote{\texttt{Bad}}$, such that 
$\poststar{\async}(\phi,\config)\cap\denote{\texttt{Bad}} \neq\emptyset$.

 For that matter, we use several claims about $\phi'$. First observe that $\phi'$ is uniquely-sequentializable.

 \begin{myclaim}\label{claim:machine-like}
Let $\config$ be a machine-like $\phi$-configuration. Then for all configuration $\config'\in\poststar{\async}(\phi',\config)=\poststar{\sync}(\phi',\config)$,
$\config'$ is machine-like. 
 \end{myclaim}

\begin{myclaim}\label{claim:if-successor-then-config}
Let $\tuple{\config,\internalState,\Views}$ be a \stable and machine-like asynchronous $\phi$-configuration corresponding to the following
word-configuration: $\block{3}$$\free^{n_1}$$\robot_{c_1}$$\free^{n'_1}$$\block{4}$$\free^{n_2}$$\robot_{c_2}$$\free^{n'_2}$$\block{5}$$\free^m$$\robot_c$$\free^n$$\block{6}$$
\free^i$$\robot_\ell$$\free^{i'}$$\block{7}$$\free^p$$\robot_{\ell'}$$
\free^r$$\block{8}$
$\robot_{tt}$$\free$$\robot_t$$\free$$\free$$\robot_g$$\free$$\robot_d$. Let $M(\config)=(\ell_i, n_1,n_2)$ a non-halting $M$-configuration. Assume that $\config$ has the following
properties.
$(i)$ if $n_1>m$ then $n\geq n_1 - m$ (if $\ell_i$ modifies $c_1$) or if $n_2>m$ then $n\geq n_2 - m$ (if $\ell_i$ modifies $c_2$),
$(ii)$ if $\ell_i$ increments $c_1$ (resp. $c_2$) then $n'_1>0$ (resp. $n'_2>0$), and 
$(iii)$ $i+i' = p+r = |L|$.
Then $\tuple{\config,\internalState,\Views}\asynchtrans^*_{\phi'} \tuple{\config',\internalState',\Views'}$ with $\config'$
\stable and machine-like, such that $M(\config)\vdash M(\config')$. 
\end{myclaim}

\begin{myclaim}\label{claim:if-config-then-successor}
Let $\tuple{\config,\internalState,\Views}$ and $\tuple{\config',\internalState',\Views'}$ be two asynchronous configurations, with $\config$ and $\config'$ \stable.
If there exists $k>0$ such that $\tuple{\config,\internalState,\Views} \asynchtrans^k_{\phi'} \tuple{\config',\internalState',\Views'}$ and for all $0<j<k$,
if $\tuple{\config,\internalState,\Views} \asynchtrans^j_{\phi'} \tuple{\config'',\internalState'',\Views''}$ we have $\config''$ not \stable, then 
$M(\config)\vdash M(\config')$.

\end{myclaim}

Assume that $M$ halts. There is then a finite bound $K\in\nat$ on the values of the counter during the run. We show hereafter an asynchronous 
$\phi$-run leading to a collision. Let $\tuple{\config_0,\internalState_0,\Views_0}$ be the initial configuration of the run, with $\config_0$ 
a $\phi$-configuration corresponding to the word-configuration 
\begin{equation*}
\block{3}\robot_{c_1}\free^{K}\block{4}\robot_{c_2}\free^{K}\block{5}\robot_c\free^K\block{6}\robot_\ell\free^{|L|}\block{7}\robot_{\ell'}\free^{|L|}\block{8}\robot_{tt}\free\robot_t\free\robot_g\free\robot_d,
\end{equation*} 
hence such that $M(\config_0)=(\ell_0,0,0)$, and $\internalState_0(i)=\lo$ 
for all $i\in\robots$. Then $\tuple{\config_0,\internalState_0,\Views_0}\asynchtrans \tuple{\config_1,\internalState_1,\Views_1}$ where
$\config_1=\config_0$, $\internalState_1(g)=\mo$, $\Views_1(g)=\ViewR{\config_0}{g}$, and $\internalState_1(i)=\internalState_0(i)$ and $\Views_1(i)=\Views_0(i)$
for all $i\in\robots\setminus\{g\}$. It is easy to check that,
at each step of a run starting from $\tuple{\config_0,\internalState_0,\Views_0}$, the conditions $(i)$, $(ii)$ and
$(iii)$ of Claim~\ref{claim:if-successor-then-config} are satisfied. Then, since $M$ halts, by applying iteratively Claim~\ref{claim:if-successor-then-config}, we have $\tuple{\config_1,\internalState_1,\Views_1}\asynchtrans^*_{\phi'} \tuple{\config_n,\internalState_n,\Views_n}$ with $M(\config_n)$ a halting configuration.
Now, the $\phi$-run continues with robot $d$ being scheduled to look, and then robots $g$
and $d$ moving, leading to a collision. Formally : 
$\tuple{\config_n,\internalState_n, \Views_n}\asynchtrans \tuple{\config,\internalState,\Views}\asynchtrans \tuple{\config',\internalState',\Views'}\asynchtrans
\tuple{\config'',\internalState'',\Views''}$, with $\config_n=\config$, $\internalState(d)=\mo$, and $\Views(d)=\ViewR{\config_n}{d}$, $\config'(d)=\config(d)-1$ since 
$\Views(d)$ memorizes the encoding of a halting configuration, with distance between $\robot_g$ and $\robot_d$ equal to 2. Last, $\config''(g)=\config(g)+1$
since $\Views(g)$ memorizes the encoding of the initial configuration, with distance between $\robot_g$ and $\robot_d$ equal to 1. Hence, $\config''(g)=\config''(d)$
and $\config''\in\poststar{as}(\phi,\config_0)\cap\denote{\texttt{Bad}}$.

Conversely assume there is an asynchronous $\phi$-run $\rho$ leading to a collision. By Claim~\ref{claim:machine-like}, if $\rho$ is in fact a $\phi'$ asynchronous
run, there is no collision (either $\rho$ starts in a non machine-like configuration and no robot moves, or it starts in a machine-like configuration and
it never reaches a collision, since machine-like configurations are collision-free by construction). Hence $\rho$ contains moves from $\robot_g$ and/or 
$\robot_d$. Assume $\robot_d$ moves in this run. Hence, there is a configuration $\tuple{\config,\internalState,\Views}$ in this run where $\robot_d$ has just
been scheduled to move and is hence such that $\internalState(\robot_d)=\mo$, $\Views(\robot_d)=\ViewR{\config}{\robot_d}$, with 
$\mov(\phi,\Views(\robot_d))\neq 0$. Hence, by definition of $\phi$, $\config$ is machine-like, \stable and $M(\config)$ is a halting $M$-configuration. Let $j\neq
\robot_g$ a robot such that $\internalState(j)=\mo$. If $\mov(\phi,V(j))\neq\{0\}$, by Proposition (P) from the proof of Theorem~\ref{th:runs-synch-asynch}
and since $\mov(\phi,V(j))=\mov(\phi',V(j))$, $V(j)=\ViewR{\config}{j}$, and it is impossible since $M(\config)$ is a halting configuration. Hence, from 
$\tuple{\config,\internalState,\Views}$ the only robots that can move are $\robot_d$ and $\robot_g$. If $\robot_g$ does not move, then all the following
configurations in the run is $\config'$ with $\config'(i)=\config(i)$ for all $i\neq\robot_d$ and $\config'(\robot_d)=\config(\robot_d)+1$. Hence, $\rho$ does not
lead to a collision.

Assume now that $\robot_g$ moves in $\rho$. Let $\rho'\cdot\tuple{\config_1,\internalState_1,\Views_1}\tuple{\config_2,\internalState_2,\Views_2}$ be a prefix of $\rho$ 
such that $\config_2(\robot_g)=\config_1(\robot_g)-1$ and $\config_2(i)=\config_1(i)$ for all $i\neq\robot_g$, and $\rho'\cdot\tuple{\config_1,\internalState_1,\Views_1}$ 
is a prefix of a $\phi'$-run. Then $\config_1$ is
a machine-like configuration and again by Proposition (P) of the proof of Theorem~\ref{th:runs-synch-asynch}, there is at most one robot $i\notin \{\robot_g,\robot_d\}$
such that $\mov(\phi',\Views_1(i))=\mov(\phi',\Views_2(i))\neq\{0\}$. Hence if $\robot_d$ does not move in $\rho$ then $\rho= \rho'\cdot\tuple{\config_1,\internalState_1,\Views_1}\tuple{\config_2,\internalState_2,\Views_2}\tuple{\config_3,\internalState_3,\Views_3}\cdot\rho''$ with $\config_3(i)=\config_2(i)$ for all
 the robots $i$ but possibly one, and then $\config_3$ being the only configuration appearing in $\rho''$. According to $\phi'$, $\config_3$ is collision free,
 so such a run could not yield a collision.
 
 Hence, we know that in $\rho$ both $\robot_g$ and $\robot_d$ moves. Moreover, from the above reasonings we deduce that in $\rho$, at some point $\robot_g$
 has been scheduled to look, then later on, $\robot_d$ has been scheduled to look, and just after, $\robot_g$ and $\robot_d$ have moved
 provoking a collision. Formally, $\rho$ is in the following form:
 $\rho=\rho'\cdot\tuple{\config_0,\internalState_0,\Views_0}\tuple{\config_1,\internalState_1,\Views_1}\cdot\rho''\cdot\tuple{\config_2,\internalState_2,\Views_2}
 \tuple{\config_3,\internalState_3,\Views_3}\tuple{\config_4,\internalState_4,\Views_4}\tuple{\config_5,\internalState_5,\Views_5}$ with $\rho',\rho''$ asynchronous $\phi'$-runs,
 $\internalState_0(\robot_g)=\lo$, $\internalState_1(\robot_g)=\mo$, and for all other robots $i$, $\internalState_0(i)=\internalState_1(i)$, and $\config_0=\config_1$, and 
 $\mov(\phi,\Views_1(\robot_g))\neq\{0\}$,
 $\config_2=\config_3$, $\internalState_2(\robot_d)=\lo$, $\internalState_3(\robot_d)=\mo$, $\mov(\phi,\Views_3(\robot_d))\neq\{0\}$, and either $\config_4(\robot_d)=\config_3(\robot_d)+1$, and for all $i\neq \robot_d$,
 $\config_3(i)=\config_4(i)$ and $\config_5(\robot_g)=\config_4(\robot_g)-1$ and for all $i\neq \robot_g$, $\config_4(i)=\config_5(i)$, or $\config_4(\robot_g)=\config_3(\robot_g)-1$,
  and for all $i\neq \robot_g$,
 $\config_3(i)=\config_4(i)$ and $\config_5(\robot_d)=\config_4(\robot_d)-1$ and for all $i\neq \robot_d$, $\config_4(i)=\config_5(i)$.
In both cases, we deduce that since $\mov(\phi,\Views_1(\robot_g))\neq\{0\}$, then $\config_0$ is a machine-like \stable configuration such that $M(\config_0)=C_0$,
and since $\mov(\phi,\Views_3(\robot_d))\neq\{0\}$, then $\config_2$ is a machine-like \stable configuration encoding a halting $M$-configuration $C_h$. Hence, from 
Claim~\ref{claim:if-config-then-successor}, since $\config_0\asynchtrans^*\config_2$ then $C_0\vdash\cdots \vdash C_h$ and $M$ halts.


\section{Proof of Theorem~\ref{th:universal-undecidable}}
The proof relies on a reduction from the repeated reachability problem of a deterministic three-counter zero-initializing bounded-strongly-cyclic machine $M$, 
which is undecidable \cite{Mayr03}. A counter machine is zero-initializing if from the initial instruction $\ell_0$ it first sets all the counters to 0. Moreover, an
infinite run is said to be \emph{space-bounded} if there is a value $K\in\nat$ such that all the values of all the counters stay below $K$ during the run. A counter
machine $M$ is
bounded-strongly-cyclic if every space-bounded infinite run starting from any configuration visits $\ell_0$ infinitely often.
The repeated reachability problem we consider is expressed as follows: given a 3-counter zero-initializing bounded-strongly-cyclic machine $M$, 
does there exist an infinite
space-bounded run of $M$?
A configuration of $M$ is encoded in the same fashion than in the proof of Theorem~\ref{th:undecidable-safe}, with 3 robots encoding the values
of the counters. 
\iflong A machine-like configuration in that case is of the form $\block{3}\free^{n_1}\robot_{c_1}\free^*\block{4}\free^{n_2}\robot_{c_2}\free^*\block{5}\free^{n_3}\robot_{c_3}\free^*\block{6}\free^m\robot_c\free^n\block{7}\free^i\robot_\ell\free^{i'}\block{8}\free^p\robot_{\ell'}
\free^r\block{9}\robot_{tt}\free\robot_t\free\free$.\fi
A transition of $M$ is simulated by the algorithm in the same way than above \emph{except that if a counter is to be
increased, the corresponding robot moves accordingly even if there is no room to do it, yielding a collision}. 
\iflong Hence, 
in any machine-like non \stable configuration, exactly one robot moves (hence the only finite runs are either ending in 
a halting configuration, or in a collision, which is not machine-like).
The algorithm that governs the robots in that case is called $\overline{\phi}$ and is a variant of $\phi'$. \comment[ts]{detail more the algorithm?}

Let $\texttt{Machine\_like}$, $\texttt{Halting}$ and $\texttt{Collision}$ be three QFP formulae, with 50 free variables, such that 
$\config\in\denote{\texttt{Machine\_like}}$ 
(respectively $\config\in\denote{\texttt{Halting}}$, $\config\in\denote{\texttt{Collision}}$) if
and only if $\config$ is machine-like (respectively iff $M(\config)$ -- the $M$-configuration encoded by $\config$ -- is a halting configuration, and iff $\config(i)=\config(j)$ for some $i,j\in\robots$).
We then let $\texttt{Goal}=\neg\texttt{Machine\_like}\vee \texttt{Halting}\vee \texttt{Collision}$ and $\texttt{Ring(y)}=y\geq 0$.

We \iflong now \else can \fi show that there is a size $n\in\denote{\texttt{Ring}}$, and a $(50,n)$-configuration $\config$ such that $\poststar{\sync}(\phi,\config)\cap\denote{\texttt{Goal}}=\emptyset$ if and only 
if there exists an infinite space-bounded run of $M$. From
Theorem~\ref{th:runs-synch-asynch}, this \iflong also provides \else provides \fi an undecidability proof for 
\reachability{\ssync} and
\reachability{\async}.
\fi

We use the following claims, reminiscent of the claims used in the proof of Theorem~\ref{th:undecidable-safe}:

\begin{myclaim}\label{cl:pos-machine-like}
Let $\config$ be a machine-like configuration. Then, for all configuration $\config'\in\poststar{\sync}(\overline{\phi},\config)$, $\config'\in
\denote{\texttt{Machine\_like}\vee\texttt{Collision}}$.
\end{myclaim}

\begin{myclaim}\label{cl:successor-then-config}
Let $\config$ be a \stable and machine-like synchronous $\overline{\phi}$-configuration corresponding to the following
word-configuration: $\block{3}\free^{n_1}\robot_{c_1}\free^{n'_1}\block{4}\free^{n_2}\robot_{c_2}\free^{n'_2}\block{5}\free^{n_3}\robot_{c_3}\free^{n'_3}\block{6}
\free^m\robot_c\free^n\block{7}\free^i\robot_\ell\free^{i'}\block{8}$\\
$\free^p\robot_{\ell'}
\free^r\block{9}\robot_{tt}\free\robot_t\free\free.$ Let $M(\config)=(\ell_i, n_1,n_2,n_3)$ be a non-halting $M$-configuration. Assume that $\config$ has the following
properties.
$(i)$ if $n_1>m$ (respectively $n-2>m$, $n-3>m$), then $n\geq n_1 - m$ (respectively $n\geq n_2-m$, $n\geq n_3-m$)(if $\ell_i$ modifies $c_1$ - respectively
$c_2$ or $c_3$), 
$(ii)$ if $\ell_i$ increments $c_1$ (resp. $c_2$ or $c_3$) then $n'_1>0$ (resp. $n'_2>0$, or $n'_3>0$), and 
$(iii)$ $i+i' = p+r = |L|$.
Then $\config\asynchtrans^*_{\overline{\phi}} \config'$ with $\config'$
\stable and machine-like, such that $M(\config)\vdash M(\config')$.
\end{myclaim}

\begin{myclaim}\label{cl:config-then-successor}
Let $\config, \config'$ be two \stable, machine-like configurations. If there exists some $k>0$ such that $\config\synchtrans_{\overline{\phi}}^k \config'$ and that for all $0<j<k$, 
if $\config\synchtrans_{\overline{\phi}}^j\config''$ then $\config''$ is not \stable, then $M(\config)\vdash M(\config')$.
\end{myclaim}

\begin{myclaim}\label{cl:stable}
Let $\config$ be a machine-like configuration, which is not \stable. Then, either $|\poststar{\sync}(\overline{\phi},\config)|$ is finite, or there exists 
$\config'\in\poststar{\sync}(\overline{\phi},\config)$ with $\config'$ \stable.
\end{myclaim}

If there is an infinite space-bounded run of $M$, we let $K\in\nat$ be the maximal values of all the counters during this run. Let $\config_0$ be the $\phi$-configuration
having the following word-representation:
\noindent$\block{3}\robot_{c_1}\free^{K}\block{4}\robot_{c_2}\free^{K}\block{5}\robot_{c_3}\free^K\block{6}\robot_c\free^K\block{7}\robot_\ell\free^{i'}\block{8}
\free^p\robot_{\ell'}\free^r\block{9}$\\
$\robot_{tt}\free\robot_t\free\robot_g\free\robot_d$. Hence $M(\config_0)$ is the initial configuration of $M$. It is easy to show that, for all 
$\config\in\poststar{\sync}(\overline{\phi}, \config_0)$, $\config$ satisfies conditions $(i)$, $(ii)$, and $(iii)$ of Claim~\ref{cl:successor-then-config}. Hence, by
applying iteratively Claim~\ref{cl:successor-then-config}, we can
build an infinite  $\overline{\phi}$-run $\rho=\config_0\config_1\cdots$
such that $\config_i\notin\denote{\texttt{Collision}\vee \texttt{Halting}}$ for all $i\geq 0$. Hence, by Claim~\ref{cl:pos-machine-like}, $\config_i\in\denote{\texttt{Machine\_like}}$
 for all $i\geq 0$, and $\config_i\notin\denote{\texttt{Goal}}$ for all $i\geq 0$. Moreover, for all machine-like configuration $\config$, for all $i\in\robots$, 
$\ViewR{\config}{i}\neq\ViewL{\config}{i}$, then it can have at most one successor, and $\rho$ is the only $\phi$-run starting from $\config_0$. Hence,
$\poststar{\sync}(\phi,\config_0)\cap\denote{\texttt{Goal}}=\emptyset$.

Conversely, assume that there is $n\in\nat$ and a $(k,n)$-$\phi$-configuration $\config_0$ such that 
$\poststar{\sync}(\overline{\phi},\config_0)\cap\denote{\texttt{Goal}}=\emptyset$. Hence, for all $\config\in\poststar{\sync}(\overline{\phi},\config_0)$,
$\config\in\denote{\texttt{Machine\_like}}$. Then, according to the definition
of the protocol, there is a unique synchronous $\overline{\phi}$-run starting from $\config_0$. Assume for the sake of contradiction that this run is finite. 
Let $\rho=\config_0\cdots\config_m$ be such a run. Then $\config_m\notin\denote{\texttt{Goal}}$ hence $M(\config_m)$ is not a halting configuration and $\config_m$
is collision-free. If it is not \stable, $\post_{\sync}(\overline{\phi},\config_m)\neq\emptyset$, from the definition of the protocol. If $\config_m$ is \stable, either 
$M(\config_m)=(\ell_h,n_1,n_2,n_3)$ for some $n_1,n_2,n_3 \in\nat$, but then $\config_m\in\denote{\texttt{Halting}}\subseteq\denote{\texttt{Goal}}$, which is not 
possible, or there exists a configuration $\config_{m+1}$ such that $\config_m\synchtrans\config_{m+1}$ and $\rho$ can be continued. Hence let $\rho=\config_0\config_1\cdots$ be the infinite synchronous $\overline{\phi}$-run
starting from $\config_0$. Assume that $\config_0$ is not \stable. Then, by Claim~\ref{cl:stable}, there exists $i\geq 0$, such that $\config_i$
is \stable. By definition of $\overline{\phi}$, $\config_{i+1}$ is not \stable, but by Claim~\ref{cl:stable} and Claim~\ref{cl:config-then-successor}, there exists $j>i+1$ such that
$\config_j$ is \stable and $M(\config_i)\vdash M(\config_j)$. By iterating this reasoning, we can in fact build an infinite run of $M$ starting in $M(\config_i)$. 
Let $K$ be the maximal number of positions between respectively $\block{3}$ and $\block{4}$, $\block{4}$ and $\block{5}$, $\block{5}$ and $\block{6}$
and $\block{6}$ and $\block{7}$ in $\config_0$. It is easy to see that this distance is an invariant of any $\overline{\phi}$-run. Hence, for any $k\geq 0$ such that $\config_k$
is \stable, $M(\config_k)=(\ell, n_1,n_2,n_3)$ with $n_i\leq K$ for $i\in\{1,2,3\}$, and the infinite run of $M$  is indeed space-bounded. Let $C_0\vdash C_1\cdots$
be such a run. Since $M$ is 
bounded-strongly-cyclic, there exists $i\geq 0$ such that $C_i=(\ell_0,n_1,n_2,n_3)$ with $n_i\in \nat$ for $i=\{1,2,3\}$, and since $M$ is zero-initializing, then 
there exists $j\geq i$ such that $C_j=(\ell_0,0,0,0)$. Hence,
$M$ has an infinite space-bounded run from $(\ell_0,0,0,0)$.

\section{Correction of the Presburger formulae}
\begin{lemma}\label{prop:configview}
 For all $i\in [1,k]$, we have  $n,\config,\mathbf{V} \models \configview{i}$ 
if and only if $\tuple{d_1,\dots, d_k} = \ViewR{\config}{i}$ on a ring of size $n$ .
\end{lemma}

\begin{proof}
Assume  $n,\config,\mathbf{V} \models \configview{i}$. Then, there exist $k-1$ variables, $d'_1,\dots, d'_{k-1}\in [1,n]$ such that 
$0<d'_1\leq d'_2\leq\dots\leq d'_{k-1}$. Moreover, there exists a bijection $f:[1,k-1]\rightarrow [1,k-1]$ such that, for all $j\neq i$,
$p_j = (p_i+d'_{f(j)}) \mod n$. Finally, $d_1=d'_1$ and for all $j\in [2,n-1]$, $d_j=d'_j-d'_{j-1}$ and $d_k=n-d'_{k-1}$. Hence, if we consider
the configuration $\config$ defined by $\config(j)=p_j$ for all $j\in[1,n]$, then $\tuple{d_1,\dots, d_k} = \ViewR{\config}{i}$.
Conversely, let $\config$ be a $(k,n)$-configuration and $\ViewR{\config}{i}=\tuple{d_1,\dots,d_k}$. Then, 
$n,\config,\mathbf{V} \models \configview{i}$. Indeed, by definition of the view, we let $d_i(j)\in [1,n]$ be such that 
$(\config(i) + d_i(j)) \modulo n = \config(j)$ for all $j\neq i$ and we let $i_1,\dots, i_k$ be a permutation of positions such that 
$0<d_i(i_1)\leq d_i(i_2)\leq\dots\leq d_i(i_{k-1})$. Then, for all $j\in [2,k-1]$, $d_j=d_i(i_j)-d_i(i_{j-1})$, $d_1=d_i(i_1)$ and $d_k=n-d_i(i_{k-1})$. By interpreting
the variables $d'_1,\dots, d'_k$ by respectively $d_i(i_1), \dots, d_i(i_k)$, it is easy to see that the formula is satisfied.
\end{proof}

\begin{lemma}\label{prop:viewsym}
For all $n\in\nat$, for all views $\mathbf{V},\mathbf{V}'\in [0,n]^k$,  we have $\mathbf{V},\mathbf{V'} \models \viewsym$ if and only if
 $\mathbf{V}'=\revert{\mathbf{V}}$.
\end{lemma}

\begin{proof}
Given $n\in\nat$, $d_1,\dots, d_k, d'_1,\dots, d'_k\in [0,n]$ such that $d_1\neq 0$ we have $\tuple{d'_1,\dots, d'_k}=\revert{\tuple{d_1,\dots,d_k}}$
if and only if there exists $1\leq j\leq k$ such that $d_\ell=0$ for all $j+1\leq \ell\leq k$ and $d'_1=d_j$, \dots, $d'_j=d_1$ and $d'\ell=0$ for all $j+1\leq \ell\leq k$
(by definition), if and only if $d_1,\ldots,d_k,d'_1,\ldots,d'_k\models \viewsym$.
\end{proof}

\begin{lemma}\label{prop:presburgmove}
For all $n\in\nat$ and a $(k,n)$-configurations $\config$ and $p' \in [0,n-1]$, we have $n,\config,\config' \models \presburgmove^\phi_i$ if and only if $p'=(\config(i)+m)\modulo n$ with $m\in\mov(\phi,\ViewR{\config}{i})$.
\end{lemma}

\begin{proof}
We have $n,\config,\config' \models \presburgmove^\phi_i$ if and only if there exist $d_1,\dots, d_k, d'_1, \dots, d'_k\in [0,n]$ such that
$\tuple{d_1,\dots, d_k} = \ViewR{\config}{i}$ (by Lemma~\ref{prop:configview}) and $\tuple{d'_1,\dots,d'_k}=\revert{\tuple{d_1,\dots,d_k}}=\ViewL{\config}{i}$
(by Lemma~\ref{prop:viewsym}) and either $(i)$ $\ViewR{\config}{i}\models\phi$ and $p'=(p_i+1) \modulo n$, or $(ii)$ $\ViewL{\config}{i}\models\phi$ and
$p'=(p_i-1)\modulo n$, or $(iii)$ $\ViewR{\config}{i}\not\models\phi$, $\ViewL{\config}{i}\not\models\phi$ and $p'=p_i$, if and only if 
either $(i)$ $1\in \mov(\phi, \ViewR{\config}{i})$ and $p'=(p_i+1)\modulo n$ or $(ii)$ $-1\in\mov(\phi,\ViewR{\config}{i})$ and $p'=(p_i-1)\modulo n$ or 
$(iii)$ $\mov(\phi,\ViewR{\config}{i})=\{0\}$ and $p'=p_i$ if and only if $p'=(p_i+m)\modulo n$ with $m\in\mov(\phi,\ViewR{\config}{i})$.
\end{proof}


\end{document}